\documentclass[10pt,letterpaper]{llncs}
\usepackage{amsmath,amssymb,url,ifthen}

\DeclareMathAlphabet{\mathpzc}{OT1}{pzc}{m}{it}

\newcommand{\A}{\mathcal{A}}
\newcommand{\B}{\mathcal{B}}

\newcommand{\G}{\mathbb{G}}
\newcommand{\Z}{\mathbb{Z}}
\newcommand{\N}{\mathbb{N}}

\newcommand{\PR}{\mathrm{Pr}}
\newcommand{\ID}{\mathsf{ID}}
\newcommand{\U}{\mathsf{U}}

\renewcommand{\check}{\stackrel{?}{=}}

\def \sample { \stackrel{\begin{footnotesize} _{\$}
\end{footnotesize}}{\leftarrow} }

\title{Towards Black-Box Accountable Authority IBE with Short Ciphertexts and Private Keys}

\author{Beno\^it Libert$^{1}$ \and
        Damien Vergnaud$^{2}$ \thanks{The first author
acknowledges  the Belgian National Fund for Scientific Research
(F.R.S.-F.N.R.S.) for their financial support  and the BCRYPT
Interuniversity Attraction Pole. The second author is supported by
the European Commission through the IST Program under Contract
ICT-2007-216646 ECRYPT II and by the French \emph{Agence Nationale de
la Recherche} through the PACE project.
 }  \thanks{This is the full
version of a paper with the same title  presented   in Public Key
Cryptography 2009 \cite{LV09} }  }
\institute{Universit\'e Catholique de Louvain,   Microelectronics Laboratory  \\
Place du Levant, 3 -- 1348 Louvain-la-Neuve -- Belgium \\ \and
Ecole Normale Sup\'erieure  -- C.N.R.S. -- I.N.R.I.A.  \\
45, Rue d'Ulm --  75230 Paris CEDEX 05 -- France  }

\begin{document}

\maketitle

\begin{abstract}
At Crypto'07, Goyal introduced the concept of \emph{Accountable
Authority Identity-Based Encryption}   as a convenient tool to
reduce the amount of trust in authorities in Identity-Based
Encryption. In this model, if the Private Key Generator (PKG)
maliciously re-distributes users' decryption keys, it runs the risk
of being caught and prosecuted. Goyal proposed two constructions:
the first one is efficient but can only trace well-formed decryption
keys to their source; the second one allows tracing   obfuscated
decryption boxes in a model (called {\it weak black-box} model)
where cheating authorities have no decryption oracle. The latter
scheme is unfortunately  far less efficient in terms of decryption
cost and ciphertext size. The contribution of this paper is to
describe a new construction that combines the efficiency of
  Goyal's first proposal with a very simple weak black-box tracing mechanism. The proposed    scheme is presented in the
selective-ID  model but readily
extends to meet all security properties in the adaptive-ID sense, which is not known to be true for prior black-box   schemes. \\

\end{abstract}

\section{Introduction}

Identity-based cryptography, first proposed by Shamir \cite{Sha},
alleviates the need for digital certificates used in traditional
public-key infrastructures. In such systems, users'  public keys are
public identifiers ({\it e.g.} email addresses) and the matching
private keys are derived by a trusted party called Private Key
Generator (PKG). The first practical construction for
\emph{Identity-Based Encryption} (IBE) was put forth by Boneh and
Franklin \cite{BF2} -- despite the
  bandwidth-demanding proposal by Cocks \cite{Co} -- and,
since then, a  large body of work has been devoted to the design of
schemes with
   additional properties or relying on
 different algorithmic  assumptions
\cite{GS0,BB2,BB3,SW05,Wat05,BBG05,Ge06,BW06,BGH07}. \\
\indent In spite of its appealing  advantages, identity-based
encryption has not undergone rapid adoption as a standard. The main
reason is arguably the fact that it requires unconditional trust in
the PKG: the latter can indeed decrypt any ciphertext or, even
worse, re-distribute users' private keys. The key escrow problem can
be mitigated  as suggested in \cite{BF2}  by sharing the master
secret among multiple PKGs, but this inevitably entails extra
communication and infrastructure. Related paradigms \cite{Ge03,AP03}
strived to remove the key escrow problem but only did so at the
expense of losing the benefit of human-memorizable public keys:
these models get rid of escrow authorities but both involve
traditional (though not explicitly
certified) public keys that are usually less convenient to work with than easy-to-remember public identifiers. \\
\indent In 2007, Goyal \cite{Goyal} explored a new approach to
 deter rogue actions from  authorities. With the
    \emph{Accountable Authority Identity-Based Encryption} (A-IBE) primitive, if the
    PKG
discloses a decryption key associated with some identity over the
Internet, it runs the risk of being caught and sued by the user.
A-IBE schemes achieve this goal by means of an interactive  private
key generation protocol between the user and the PKG. For each
identity,
 there are exponentially-many families of possible decryption
 keys.
  The key generation
protocol provides the user with a single decryption key   while
concealing to the PKG the family that  this key belongs to. From
this  private key,  the user is computationally unable to find one
from a different family. Hence, for a given identity, a pair of
private keys from distinct families   serves as evidence of a
fraudulent PKG.
 The latter
remains able to passively eavesdrop
 communications  but is discouraged to   reveal users' private keys.
Also, users cannot falsely accuse an honest PKG since they are
unable to compute a new   key from a different family using a given
key.

\medskip \noindent \textsc{Prior Works.}  Two
constructions were given in \cite{Goyal}. The first one (that we
call $\mathpzc{Goyal}$-$\mathpzc{1}$ hereafter)
 builds  on Gentry's IBE \cite{Ge06} and, while efficient,     only allows tracing well-formed decryption keys. This white-box model seems unlikely to suffice in practice  since
malicious parties can rather release
 an imperfect and/or
 obfuscated   program that only decrypts with small but noticeable probability.    The second scheme of \cite{Goyal} (let us call it
$\mathpzc{Goyal}$-$\mathpzc{2}$), which is constructed on the
Sahai-Waters fuzzy IBE \cite{SW05}, has a variant  providing weak
black-box traceability: even an imperfect pirate   decryption box
can be traced  (based on its input/output behavior) back to
  its source although traceability is only guaranteed against dishonest
 PKGs that have no decryption oracle in the attack game. However, $\mathpzc{Goyal}$-$\mathpzc{2}$
  is somewhat inefficient as decryption requires a number of
pairing
 calculations that is linear in the security parameter. For the usually required  security level, ciphertexts contain more than $160$ group
 elements and decryption    calculates
 a product of about $160$ pairings. \\ \indent
  Subsequently, Au {\it et al.} \cite{AHLSW} described another A-IBE scheme providing retrievability ({\it i.e.}, a property that prevents  the PKG from revealing more than one key for a given identity without exposing
 its master key) but remained in  the white-box model. More recently, Goyal {\it et al.} \cite{GLSW}
   modified the $\mathpzc{Goyal}$-$\mathpzc{2}$ system   using attribute-based encryption techniques \cite{SW05,GPSW}
   to achieve full black-box traceability: unlike   $\mathpzc{Goyal}$-$\mathpzc{2}$, the scheme of  \cite{GLSW} preserves security
   against dishonest PKGs that have access to a decryption oracle in the model.
  While definitely  desirable in practice, this property is currently  achievable only at the expense of the same significant penalty as
 in $\mathpzc{Goyal}$-$\mathpzc{2}$
  \cite{Goyal} in  terms of  decryption
  cost and ciphertext  size.

\medskip\noindent \textsc{Our Contributions.} We present a    very efficient
and  conceptually simple scheme with weak black-box traceability. We
prove its security  (in the standard model) under the
 same   assumption as $\mathpzc{Goyal}$-$\mathpzc{2}$.  Decryption keys and
ciphertexts consist of a constant number of group elements and their
length is thus linear in the security parameter $\lambda$ (instead
of quadratic as in $\mathpzc{Goyal}$-$\mathpzc{2}$). Encryption and
decryption take $O(\lambda^3)$-time (w.r.t. $O(\lambda^4)$ in
$\mathpzc{Goyal}$-$\mathpzc{2}$) with only two pairing computations
as for the latter (against more than $160$ in
$\mathpzc{Goyal}$-$\mathpzc{2}$).
\\ \indent While presented in the selective-ID security model (where
adversaries must choose the identity that will be their prey at the
outset of the game) for simplicity, our scheme is easily adaptable
  to the adaptive-ID model of \cite{BF2}. In contrast,
one of the security properties ({\it i.e.}, the infeasibility for
users to frame innocent PKGs) was only established in the
selective-ID setting for known schemes in the black-box model ({\it
i.e.}, $\mathpzc{Goyal}$-$\mathpzc{2}$ and its fully black-box
extension \cite{GLSW}). Among such schemes, ours thus appears to be
the first one that can be tweaked so as to achieve   adaptive-ID
security against dishonest
users.  \\
\indent Our scheme performs almost as well as
$\mathpzc{Goyal}$-$\mathpzc{1}$ (the main overhead being a long
master public key {\it \`a la} Waters \cite{Wat05} to obtain the
adaptive-ID security). In comparison with the latter, that was only
analyzed in a white-box model of traceability, our system provides
several other advantages:
\begin{itemize}
\item[-]  Its security relies on a weaker assumption. So far, the
only fully practical A-IBE scheme
 was resting on
      assumptions      whose strength grows   with the number of adversarial queries, which can be as large as $ 2^{30}$ as
commonly assumed in the literature. Such assumptions are subject to
a limited   attack \cite{Cheon} that requires a careful adjustment
of group sizes (by as much as $50\%$ additional bits) to guarantee a
secure use of   schemes.
\item[-] It  remains secure when many users
want to run the key generation protocol in a concurrent fashion.
$\mathpzc{Goyal}$-$\mathpzc{1}$     has a key generation protocol
involving zero-knowledge proofs. As its security reductions require
to rewind adversaries at each key generation query, security is only
guaranteed when the PKG interacts with users sequentially. In
inherently concurrent environments like the Internet,   key
generation protocols should remain secure when  executed by many
users willing to register at the same time. By minimizing the number
of rewinds in
  reductions, we   ensure that our scheme
remains secure  in a concurrent setting. In these regards, the key
generation protocol of $\mathpzc{Goyal}$-$\mathpzc{2}$ makes use of
oblivious  transfers (OT)  in sub-protocols. It thus supports
concurrency whenever the underlying OT protocol does. As already
mentioned  however, our
  scheme features a much better efficiency than $\mathpzc{Goyal}$-$\mathpzc{2}$.
\item[-] In a white-box model of traceability, it can be made  secure against
dishonest PKGs equipped with a decryption oracle\footnote{We believe
that the $\mathpzc{Goyal}$-$\mathpzc{1}$ system can also be modified
so as to obtain this property. }. In the following, we nevertheless
focus on the (arguably more interesting) weak black-box traceability
aspect.
\end{itemize}

As an extension to the proceedings version of this paper
\cite{LV09}, we also show how to apply the idea of our weak
black-box tracing mechanism to Gentry's IBE. The resulting A-IBE
system is obtained by bringing a simple modification to the key
 generation protocol of $\mathpzc{Goyal}$-$\mathpzc{1}$ so as to perfectly
 hide the user's key family from the PKG's view while preserving the efficiency of the whole
 scheme. Since the resulting system inherits the efficiency of Gentry's IBE and the $\mathpzc{Goyal}$-$\mathpzc{1}$ white-box A-IBE,
 it
   turns out to be the most efficient
 weakly black-box A-IBE construction to date. Its (adaptive-ID)
 security is moreover proved under a tight reduction (albeit under a strong
 assumption). \\
 \indent
Finally, since detecting misbehaving PKGs is an equally relevant
problem in IBE primitives and their generalizations, we show how the
underlying idea of previous schemes can be applied to one of the
most practical identity-based broadcast encryption (IBBE)
realizations \cite{BH08}. We also argue that the same technique
similarly applies in the context of attribute-based encryption
\cite{SW05,GPSW}.

\medskip
\noindent\textsc{Organization.} In the rest of the paper, section
\ref{model} recalls the A-IBE security model defined in
\cite{Goyal}. We first analyze the white-box version of our scheme
in section \ref{efficient-DBDH} and then describe a weak black-box
tracing mechanism in section \ref{weak-BB}. Sections \ref{Gentry}
and \ref{IBBE} describe and analyze the extensions of our method to
Gentry's IBE and the Boneh-Hamburg IBBE scheme, respectively.

\section{Background and Definitions}
\label{model}
\noindent \textsc{Syntactic definition and security model.} We
recall the  definition of A-IBE schemes and their security
properties as defined  in \cite{Goyal}.
\begin{definition}
An \emph{Accountable Authority Identity-Based Encryption scheme}
(A-IBE) is a tuple
($\mathbf{Setup},\mathbf{Keygen},\mathbf{Encrypt},\mathbf{Decrypt},\mathbf{Trace})$
of   efficient algorithms or protocols such that:
\begin{itemize}
\item $\mathbf{Setup}$    takes as input a security parameter and outputs    a master public
key
$\mathsf{mpk}$ and a matching \emph{master secret key} $\mathsf{msk}$.
\item $\mathbf{Keygen}^{(\mathrm{PKG},\mathsf{U})}$ is an interactive protocol between the public parameter generator
$\mathsf{PKG}$ and the user $\mathsf{U}$:
\begin{itemize}
\item[$\cdot$] the common input to $\mathrm{PKG}$ and $\mathsf{U}$ are:  the master public key $\mathsf{mpk}$ and an
identity $\mathsf{ID}$  for which the decryption key has to be
generated;
\item[$\cdot$] the private input to $\mathrm{PKG}$ is
the master secret key $\mathsf{msk}$.
\end{itemize}
Both parties may use a sequence of private
  coin tosses as additional inputs. The protocol ends with
$\mathsf{U}$ receiving a  decryption key  $d_\mathsf{ID}$ as his
private output.
\item $\mathbf{Encrypt}$  takes as input the master public key $\mathsf{mpk}$, an identity $\mathsf{ID}$ and a message $m$
 and outputs a \emph{ciphertext}.
\item $\mathbf{Decrypt}$   takes as input  the master public key $\mathsf{mpk}$, a decryption key $d_\mathsf{ID}$ and a ciphertext $C$ and outputs a message.
\item $\mathbf{Trace}$ given    the master public key $\mathsf{mpk}$, a decryption key $d_\mathsf{ID}$,
this algorithm  outputs a \emph{key family number} $n_F$ or the
special symbol $\perp$ if $d_\mathsf{ID}$ is ill-formed.
\end{itemize}
Correctness  requires that, for any outputs
$(\mathsf{mpk},\mathsf{msk})$ of $\mathbf{Setup}$, any plaintext $m$
and any identity $\mathsf{ID}$, whenever $d_{\ID} \leftarrow
\mathbf{Keygen}^{(\mathrm{PKG}(\mathsf{msk}),\mathsf{U})}(\mathsf{mpk},\mathsf{ID})$,
we have
$$\begin{array}{c}
 \mathbf{Trace}\big(\mathsf{mpk}, d_{\ID} \big) \neq
\perp, \\
\mathbf{Decrypt}\big(\mathsf{mpk},d_{\ID},\mathbf{Encrypt}(\mathsf{mpk},\mathsf{ID},m)\big)
= m.
\end{array}
$$
\end{definition}

The above definition  is for the white-box setting. In a black-box
model,
  $\mathbf{Trace}$ takes as input an identity $\ID$, the
corresponding user's well-formed private key $d_{\ID}$ and a
decryption box $\mathbb{D}$ that successfully opens a non-negligible
fraction $\varepsilon$ of ciphertexts encrypted under $\ID$. The
output of $\mathbf{Trace}$ is either ``PKG'' or ``User'' depending
on which party is found guilty for having crafted $\mathbb{D}$.
\\ \indent  Goyal
formalized three security properties for A-IBE schemes. The first
one is the standard notion of privacy \cite{BF2} for IBE systems. As
for the
  other ones, the \textrm{FindKey} game captures the
intractability for the PKG to create a decryption key of the same
family as the one obtained by the user during the key generation
protocol. Finally, the \textrm{ComputeNewKey} game models the
infeasibility for users to generate a   key $d_\mathsf{ID}^{(2)}$
outside the family of the legally obtained one
$d_\mathsf{ID}^{(1)}$.

\begin{definition}\label{sec-def-A-IBE} An A-IBE scheme is  deemed secure if
all probabilistic polynomial time (PPT) adversaries have negligible
advantage in the following   games.
\end{definition}
\begin{enumerate}
\item \textbf{The \textrm{IND-ID-CCA} game.} For  any PPT algorithm $\A$, the model considers the following game, where $\lambda \in \N$ is a security parameter:
\begin{center}
\begin{tabular}{l}
\fbox{$\mathbf{Game}_{\A}^{\mathrm{IND\textrm{-}ID\textrm{-}CCA}}(\lambda)$}\\
$(\mathsf{mpk},\mathsf{msk}) \leftarrow \mathbf{Setup}(\lambda)$ \\
$(m_0,m_1,\mathsf{ID^\star},s) \leftarrow \A^{\mathsf{Dec},\mathsf{KG}}(\mathsf{find},\mathsf{mpk})$ \\
\phantom{$(m_0,m_1,\mathsf{ID}) \sample$} $\left\vert
        \begin{array}{l}
          \mathsf{Dec}:  (C,\mathsf{ID}) \\ \quad \dashrightarrow
          \mathbf{Decrypt}\big(\mathsf{mpk},\mathsf{msk},\mathsf{ID},C\big) ; \\
          \mathsf{KG}: \mathsf{ID} \dashrightarrow \mathbf{Keygen}^{(\mathrm{PKG}(\mathsf{msk}),\A)}(\mathsf{mpk},\mathsf{ID}) \\
\hspace{1.3cm} \texttt{//} \quad \mathsf{ID} \neq \mathsf{ID}^\star \\

\end{array}\right.$ \\
$d^\star \sample \{0,1\}$ \\
$C^\star \leftarrow \mathbf{Encrypt}(\mathsf{mpk},\mathsf{ID}^\star,m_{d^\star})$ \\
$d \leftarrow \A^{\mathsf{Dec},\mathsf{KG}}(\mathsf{guess},s,C^\star)$ \\
\phantom{$b \sample$} $\left\vert
        \begin{array}{l}
          \mathsf{Dec}:  (C,\mathsf{ID})  \dashrightarrow
          \mathbf{Decrypt}\big(\mathsf{mpk},\mathsf{msk},\ID,C \big) ; \\
\hspace{2cm} \texttt{//} \quad  (C,\ID) \neq (C^\star,\ID^\star)\\
          \mathsf{KG}: \mathsf{ID} \dashrightarrow \mathbf{Keygen}^{(\mathrm{PKG}(\mathsf{msk}),\A)}(\mathsf{mpk},\mathsf{ID}) \\
\hspace{1.3cm} \texttt{//} \quad \mathsf{ID} \neq \mathsf{ID}^\star \\
\end{array}\right.$ \\
      \texttt{return} $1$ \texttt{if} $d=d^{\star}$  and $0$ \texttt{otherwise}.
    \end{tabular}
  \end{center}
$\mathcal{A}$'s  advantage  is measured by
    $\mathbf{Adv}_{\mathcal{A}}^{\mathrm{CCA}}(\lambda) =
  | \Pr[\mathbf{Game}_{\mathcal{A}}^{\mathrm{CCA}}=1] - {1}/{2} |.$
\end{enumerate}

The   weaker   definition   of   chosen-plaintext  security
(IND-ID-CPA)  is formalized   in  the   same   way   in   \cite{BF2}
but $\A$ is not granted access to a decryption oracle.
\begin{enumerate}
\item[2.] \textbf{The \textrm{FindKey} game.} Let $\A$ be a PPT algorithm. We consider the following game, where $\lambda \in \N$ is a security parameter:
\begin{center}
\begin{tabular}{l}
\fbox{$\mathbf{Game}_{\A}^{\mathrm{FindKey}}(\lambda)$}\\
$(\mathsf{mpk},\ID,s_1) \leftarrow \A(\mathsf{setup},\lambda)$ \\

$(d_\mathsf{ID}^{(1)},s_2) \leftarrow \mathbf{Keygen}^{(\A(s_1),\cdot)}(\mathsf{mpk},\mathsf{ID})$ \\

$d_\mathsf{ID}^{(2)} \leftarrow \A(\mathsf{findkey},s_1,s_2)$ \\

\texttt{return} $1$ \texttt{if} $\mathbf{Trace}(\mathsf{mpk},d_\mathsf{ID}^{(1)}) =  \mathbf{Trace}(\mathsf{mpk},d_\mathsf{ID}^{(2)})$\\
\phantom{{\texttt{return}}} $0$ \texttt{otherwise}.
\end{tabular}
\end{center}
$\mathcal{A}$'s advantage is now defined
  as $\mathbf{Adv}_{\mathcal{A}}^{\mathrm{FindKey}}(\lambda) =
 \Pr[\mathbf{Game}_{\mathcal{A}}^{\mathrm{FindKey}}=1].$
\end{enumerate}
Here, the adversary $\A$ acts as a cheating PKG and the challenger
emulates the honest user. Both parties engage in a key generation
protocol where the challenger obtains a private key for an identity
$\ID$ chosen by $\A$. The latter aims at  producing a private key
corresponding to $\ID$ and belonging to the same family   as the key
obtained by the challenger in   the key generation protocol. Such a
successful dishonest
   PKG
  could disclose user  keys without being caught. \\
\indent Note that, at the beginning of the experiment, $\A$
generates $\mathsf{mpk}$ without revealing the
master key $\mathsf{msk}$ and the challenger runs a sanity check on $\mathsf{mpk}$. \\
 \indent As noted in
\cite{Goyal}, it makes sense to provide $\A$ with a decryption
oracle that undoes ciphertexts using $d_\mathsf{ID}^{(1)}$ (and
could possibly leak information on the latter's family) between
steps 2 and 3 of the game. We call  this enhanced notion FindKey-CCA
(as opposed to the weaker one which we call FindKey-CPA). \\
\indent Finally, in the black-box model, instead of outputting a new
key $d_{\ID}^{(2)}$, the dishonest PKG comes up with a decryption
box $\mathbb{D}$  that correctly decrypts ciphertexts intended for
$\ID$ with non-negligible  probability $\varepsilon$ and wins if the
tracing algorithm   returns ``User'' when run on
$d_\mathsf{ID}^{(1)}$ and with oracle access to $\mathbb{D}$.
\begin{enumerate}
\item[3.] \textbf{The \textrm{ComputeNewKey} game.} For a PPT algorithm $\A$, the model finally considers the following game:
\begin{center}
\begin{tabular}{l}
\fbox{$\mathbf{Game}_{\A}^{\mathrm{ComputeNewKey}}(\lambda)$}\\
$(\mathsf{mpk},\mathsf{msk}) \leftarrow \mathbf{Setup}(\lambda)$ \\
$(d_\mathsf{ID^\star}^{(1)},d_\mathsf{ID^\star}^{(2)},\mathsf{ID^\star}) \leftarrow \A^{\mathsf{KG}}(\mathsf{mpk})$ \\
\phantom{$(d_\mathsf{ID^\star}^{(1)},d_\mathsf{ID^\star}^{(2)},\mathsf{ID^\star})
\sample$} $\left\vert
        \begin{array}{l}
          \mathsf{KG}: \mathsf{ID} \dashrightarrow \mathbf{Keygen}^{(\mathrm{PKG}(\mathsf{msk}),\A)}(\mathsf{mpk},\mathsf{ID}) \\

\end{array}\right.$ \\

\texttt{return} $1$ \texttt{if}
$\mathbf{Trace}(\mathsf{mpk},d_\mathsf{ID^\star}^{(1)}) \neq \perp$
\texttt{and} \\ \phantom{\texttt{return} $0$ if }
$\mathbf{Trace}(\mathsf{mpk},d_\mathsf{ID^\star}^{(2)}) \notin \{\perp,\mathbf{Trace}(\mathsf{mpk},d_\mathsf{ID^\star}^{(1)})\}$ \\
\phantom{{\texttt{return}}} $0$ \texttt{otherwise}.
\end{tabular}
\end{center}
 $\mathcal{A}$'s  advantage  is
    $\mathbf{Adv}_{\mathcal{A}}^{\mathrm{ComputeNewKey}}(\lambda) =
 \Pr[\mathbf{Game}_{\mathcal{A}}^{\mathrm{ComputeNewKey}}=1].$
\end{enumerate}
The ComputeNewKey game involves an adversary interacting with a PKG
in executions of the key generation protocol and obtaining  private
keys associated with {\it distinct} identities of her choosing. The
adversary is declared successful if, for some identity that may have
been queried for key generation, she is able to find {\it two}
private keys from {\it distinct} families. Such a pair   would allow
her to
trick a judge into wrongly believing in a misbehavior of the  PKG. \\
\indent In the black-box scenario, the output of the dishonest user
consist of a key $d_{\ID^\star}^{(1)}$ and a pirate decryption box
$\mathbb{D}$ that yields the correct answer with probability
$\varepsilon$ when provided with a ciphertext encrypted for
$\ID^\star$. In this case, the adversary wins if the output of
$\mathbf{Trace}^{\mathbb{D}}(\mathsf{mpk},d_{\ID^\star}^{(1)})$ is
``PKG''. \\
 \indent In \cite{CHK03}, Canetti, Halevi and Katz suggested
relaxed notions of IND-ID-CCA  and IND-ID-CPA  security where the
adversary has to choose the target identity $\ID^\star$ ahead of
time (even before seeing the master public key $\mathsf{mpk}$). This
relaxed model, called ``selective-ID'' model  (or IND-sID-CCA and
IND-sID-CPA for short), can be naturally extended to the
ComputeNewKey  notion.

\medskip\noindent \textsc{Bilinear Maps and Complexity Assumptions.} We
    use   prime order groups   $(\G,\G_T)$
endowed with an efficiently computable map  $e: \G \times \G
\rightarrow \G_T$ such that:
\begin{enumerate}
\item[1.]   $e(g^a,h^b)=e(g,h)^{ab}$ for any $(g,h)\in \G\times \G$ and $a,b\in \mathbb{Z}$;
\item[2.]   $e(g,h)\neq 1_{\G_T}$ whenever $g,h\neq 1_{\G}$.
\end{enumerate}
In such {\it bilinear groups}, we assume the hardness of the   (now
classical) Decision Bilinear Diffie-Hellman problem that has been
widely used in the recent years.

\begin{definition} \label{DBDH-def}
Let $(\G,\G_T)$ be bilinear   groups of prime  order $p$  and
 $g \in {\mathbb G}$.
The    {\bf Decision  Bilinear Diffie-Hellman Problem} (DBDH) is to
distinguish the distributions of tuples $(g^a,g^b,g^c,e(g,g)^{abc})
$ and $ (g^a,g^b,g^c,e(g,g)^{z}) $ for random values $a,b,c,z
\sample \Z_p^*$. The advantage of a distinguisher $\B$
   is measured by
\begin{multline*}
\mathbf{Adv}^{\mathrm{DBDH}}_{\G,\G_T}(\lambda)=\big|\PR[
a,b,c\sample \Z_p^* :  \B (g^a,g^b,g^c,e(g,g)^{abc})=1]
\\- \PR[ a,b,c,z \sample \Z_p^* :  \B (g^a,g^b,g^c,e(g,g)^{z})=1]
\big|.
\end{multline*}
\end{definition}
For convenience, we   use an equivalent formulation -- called {\it
 modified} DBDH -- of the   problem which is to distinguish
$e(g,g)^{ab/c}$ from random given $(g^a,g^b,g^c)$.

\section{The Basic Scheme}
\label{efficient-DBDH}
The scheme mixes ideas from the ``commutative-blinding'' \cite{BB2}
and ``exponent-inversion'' \cite{SK}   frameworks. Private keys have
the same shape as in commutative-blinding-based schemes
\cite{BB2,BB3,Wat05,BW06}. At the same time, their first element is
a product of two terms, the first one of which is inspired from
Gentry's IBE \cite{Ge06}. \\
\indent According to a   technique   applied in \cite{Goyal},
private keys contain  a   family number $t$ that cannot be tampered
with while remaining  hidden from the PKG. This family number $t$ is
determined by combining two random values $t_0$ and $t_1$
respectively chosen by the user and the PKG in the key generation
protocol. The latter begins with the user  sending a commitment $R$
to $t_0$. Upon receiving $R$, the PKG   turns it into a commitment
to $t_0+t_1$ and uses the modified commitment  to generate a
``blinded'' private key $d_{\ID}'$. The user obtains his final
  key $d_{\ID}$  by ``unblinding'' $d_{\ID}'$ thanks to the
randomness that was used to
compute $R$.  \\
\indent A difference with $\mathpzc{Goyal}$-$\mathpzc{1}$ is that
the key family number is perfectly hidden to the PKG  and the
FindKey-CPA security is unconditional. In the key generation
protocol, the user's first message  is a perfectly hiding commitment
that comes along with a witness-indistinguishable (WI) proof of
knowledge of its opening. In   $\mathpzc{Goyal}$-$\mathpzc{1}$,
users rather send a deterministic (and thus non-statistically
hiding) commitment and knowledge of the underlying value must be
proven in zero-knowledge because a proof of knowledge of a discrete
logarithm must be simulated (by rewinding the adversary) in the
proof of FindKey-CPA security. In the present scheme,  the latter
does not rely on a specific assumption and we do not need to
simulate knowing the solution of a particular problem instance.
Therefore, we can dispense with perfectly ZK proofs and settle for a
more efficient 3-move WI proof (such as Okamoto's variant
\cite{Ok92} of Schnorr \cite{Sch})
  whereas 4 rounds are needed
  using   zero-knowledge proofs of knowledge.

\subsection{Description}

\begin{description}
\item[Setup:]
given $\lambda \in \mathbb{N}$, the PKG selects bilinear groups
$(\G,\G_T)$ of prime  order $p>2^\lambda$  with a random generator
$g \sample \G$. It chooses $h,Y,Z \sample \G$ and $x  \sample
\Z_p^*$ at random. It defines its master key as $\mathsf{msk}:=x$
and the
master public key is chosen as $\mathsf{mpk}:=(X=g^x,Y,Z,h)$.    \\
\vspace{-0.3 cm}.
\item[Keygen$^{(\mathrm{PKG},\mathsf{U})~}$:] to obtain a private key for his identity $\ID$, a user $\U$ interacts with the PKG in the following
  key generation protocol. \\
\vspace{-0.3 cm}
\begin{itemize}
\item[1.] The user $\U$ draws $t_0, \theta \sample \Z_p^*$,
provides the PKG with a  commitment $R=h^{t_0} \cdot X^\theta$  and
also gives
  an interactive   witness
indistinguishable proof of knowledge of the pair  $(t_0,\theta)$, which he retains for later use. \\
\vspace{-0.3 cm}
\item[2.] The PKG   outputs $\bot$ if the proof of knowledge fails to verify. Otherwise, it picks $r',t_1 \sample \Z_p^*$ and returns
\begin{eqnarray} \label{priv-key-1}
d_{\ID}'=(d_1',d_2',d_3')=\Big( ( Y \cdot R   \cdot h^{t_1})^{1/x}
\cdot (g^\ID \cdot Z)^{r'},~X^{r'},~t_1 \Big).
\end{eqnarray}
\item[3.]   $\U$ picks $r'' \sample \Z_p^*$ and  computes
 $d_{\ID}=({d_1'}/g^{\theta} \cdot (g^{\ID} \cdot Z)^{r''},~d_2' \cdot X^{r''} ,~{d_3'}+t_0)$
which should equal
\begin{eqnarray} \label{priv-key-2}
d_{\ID}=({d_1},{d_2},{d_3})=\Big( ( Y \cdot h^{t_0+t_1})^{1/x} \cdot
(g^\ID \cdot Z)^{r},~X^{r},~t_0+t_1  \Big)
\end{eqnarray}
where $r=r'+r''$. Then, $\U$ checks whether $d_{\ID}$  satisfies the
relation
\begin{eqnarray} \label{valid-priv-1}
e(d_1,X)&=& e(Y,g) \cdot e(h,g)^{d_3} \cdot e(g^\ID \cdot Z,d_2).
\end{eqnarray}
If so, he sets his private key as $d_{\ID}$ and the latter belongs
to the family of decryption keys identified by $n_F=d_3=t_0+t_1$. He outputs $\bot$ otherwise. \\
\vspace{-0.3 cm}
\end{itemize}
\item[Encrypt:] to encrypt   $m \in \G_T$ given $\mathsf{mpk}$ and $\ID$, choose $s \sample \Z_p^*$ and compute \\
\vspace{-0.2 cm}
$$ C= \big( C_1,C_2,C_3,C_4  \big)= \Big(  X^{s}    ,~ (g^\ID \cdot Z)^{s},~e(g,h)^s   ,~ m \cdot e(g,Y)^s    \Big). $$
\item[Decrypt:] given $C= \big(C_1,C_2,C_3,C_4  \big) $ and $d_{\ID}=({d_1},{d_2},{d_3})$,
 compute
 \begin{eqnarray} \label{decrypt}
   m=C_4 \cdot \Big( \frac{e(C_1,d_1)}{e(C_2,d_2) \cdot C_3^{d_3}}
   \Big)^{-1}
 \end{eqnarray}
 \item[Trace:] given a purported private key $d_{\ID}=({d_1},{d_2},{d_3})$ and an identity $\ID$, check the validity of $d_{\ID}$ w.r.t. $\ID$
 using relation  (\ref{valid-priv-1}). If
 valid, $d_{\ID}$ is declared as a member of the family identified
 by $n_F=d_3$.
\end{description}
The correctness of the scheme   follows from the fact that
well-formed private keys always satisfy relation
(\ref{valid-priv-1}). By raising both members of
(\ref{valid-priv-1}) to the power $s\in \Z_p^*$, we see that the
quotient of pairings
 in (\ref{decrypt}) actually equals $e(g,Y)^s$. \\
\indent The scheme features about the same efficiency as classical
  IBE schemes derived from the commutative-blinding
framework \cite{BB2}. Encryption demands no pairing calculation
since $e(g,h)$ and $e(g,Y)$ can both be cached as part of the system
parameters. Decryption requires to compute a quotient of two
pairings  which is significantly faster than two
independent pairing evaluations when optimized in the same way as modular multi-exponentiations. \\
\indent In comparison with the most efficient  standard model
scheme based on the same assumption (which is currently  the first
scheme of \cite{BB2}), the only overhead is a slightly longer
ciphertext and an extra exponentiation in $\G_T$ at both ends.

\subsection{Security} \label{sec-first-scheme}
  \noindent \textsc{Selective-ID Security.}  We first
prove the IND-sID-CPA security under the    modified  DBDH
assumption (mDBDH).
\begin{theorem} \label{IND-ID-CPA-proof}
The scheme is IND-sID-CPA under the mDBDH assumption.
\end{theorem}
\begin{proof}
We show how a simulator $\B$ can interact with a selective-ID
adversary $\A$ to solve a mDBDH instance $(T_a=g^a,T_b=g^b,T_c=g^c,T
\check e(g,g)^{ab/c}  )$. At the outset of the game, $\A$ announces
the target identity $\ID^\star$. To prepare $\mathsf{mpk}$, $\B$
chooses $\alpha, \gamma, t^* \sample \Z_p^*$ and sets $X=T_c=g^c$,
$h=T_b=g^b$, $Y=X^{\gamma} \cdot h^{-t^*}$,  and $Z=g^{-\ID^\star}
\cdot X^\alpha$. The adversary's view is simulated as follows.
\begin{description}
\item[Queries:] at any time, $\A$ may trigger an execution of the
  key generation protocol for an identity $\ID\neq \ID^\star$
of her choosing. She then supplies an element $R=h^{t_0} \cdot
X^\theta$ along with a WI proof of knowledge of $(t_0,\theta)$. The
simulator $\B$ verifies the proof  but does not need to rewind the
adversary as it can answer the query without knowing $(t_0,\theta)$.
To do so, it picks  $t_1 \sample \Z_p^*$ at random  and defines $W=Y
\cdot R \cdot h^{t_1}$, $d_3'=t_1$. Elements $d_1'$ and $d_2'$ are
generated as
 \begin{eqnarray} \label{PPK}
 (d_1',d_2')=\Big( (g^{\ID} \cdot
        Z)^{r'} \cdot W^{-\frac{\alpha}{\ID-\ID^\star}} ,~
        X^{r'} \cdot W^{-\frac{1}{\ID-\ID^\star}} \Big)
   \end{eqnarray}
using a random   $r' \sample \Z_p^*$. If we set
$\tilde{r}'=r'-\frac{w}{ c (\ID-\ID^\star)}$, where $w=\log_g(W)$,
we observe
      that $(d_1',d_2')$ has the correct distribution since
       \begin{eqnarray*}
         W^{1/c} \cdot (g^\ID \cdot
        Z)^{\tilde{r}'}    & = &    W^{1/c} \cdot
        (g^{\ID-\ID^\star} \cdot X^\alpha)^{\tilde{r}'}  \\ & =&  W^{1/c} \cdot
        ( g^{\ID-\ID^\star} \cdot X^{\alpha})^{r'} \cdot  (g^{\ID-\ID^\star})^{-\frac{w}{c (\ID-\ID^\star)}
        } \cdot X^{-\frac{w \alpha}{c(\ID-\ID^\star)}} \\
        & = & ( g^{\ID} \cdot Z)^{r'} \cdot
        W^{- \frac{\alpha}{\ID-\ID^\star}}
      \end{eqnarray*}
      and $X^{\tilde{r}'} = X^{r'} \cdot
      (g^c)^{-\frac{w}{c(\ID-\ID^\star)}}=X^{r'} \cdot
      W^{-\frac{1}{\ID-\ID^\star}}$. Finally, the ``partial private key'' $(d_1',d_2',d_3')$ is returned to $\A$.  Note that
      the above calculation can be carried out without knowing $w=\log_g(W)$ or
      the representation $(t_0,\theta)$ of $R$ w.r.t. to $(h,X)$ and $\B$  does {\it not}   need to
      rewind $\A$.

\item[Challenge:] when   the first stage is over, $\A$ outputs
  $m_0,m_1 \in \G_T$. At this point, $\B$ picks $r^\star \sample \Z_p^*$ and defines a private key
  $(d_1,d_2,d_3)=(g^\gamma \cdot X^{\alpha r^\star}, X^{r^\star},t^* )$ for the identity $\ID^*$. It flips a fair
coin $d^\star \sample \{0,1\}$ and encrypts $m_{d^\star}$ as
\begin{eqnarray*}
C_1^\star=T_a=g^a \qquad C_2^\star=T_a^{\alpha} \qquad C_3^\star=T
\qquad C_4^\star=m_{d^\star} \cdot
\frac{e(C_1^\star,d_1)}{e(C_2^\star,d_2) \cdot {C_3^\star}^{d_3}}.
\end{eqnarray*}
We   see that $(d_1,d_2,d_3)$ is a valid   key for $\ID^\star$.
Since $g^{\ID^\star} \cdot Z=X^{\alpha}=T_c^{\alpha}$ and $h=g^b$,
$C^\star=(C_1^\star,C_2^\star,C_3^\star,C_4^\star)$ is a valid
encryption of $m_{d^\star}$ (with the   exponent $s=a/c$) if
$T=e(g,g)^{ab/c}$. If $T$ is random, we have $T=e(g,h)^{s'}$ for
some random $s' \in \Z_p^*$ and thus $C_4^\star=m_{d^\star} \cdot
e(Y,g)^{s} \cdot e(g,h)^{(s-s')t^*}$, which means that $m_{d^\star}$
is perfectly hidden since $t^\star$ is independent of $\A$'s view.
\end{description}
As usual, $\B$
  outputs $1$ (meaning that $T=e(g,g)^{ab/c}$) if $\A$
successfully guesses $d'=d^\star$ and $0$ otherwise. \qed
\end{proof}
In the above proof,  the simulator does not   rewind the adversary
at any time. The scheme thus remains IND-sID-CPA in concurrent
environments, where a batch of users may want to simultaneously run
 the key generation
protocol. \\
\indent Also, the simulator knows a valid private key for each
identity. This allows using hash proof systems \cite{CrSh98,CrSh02}
as in \cite{Ge06,KV08} to secure the scheme against
chosen-ciphertext attacks. The   advantage of this approach, as
shown in appendices \ref{cca} and \ref{proof-find-cca}, is
to provide FindKey-CCA security in a white-box setting. \\
\indent Unlike the  $\mathpzc{Goyal}$-$\mathpzc{1}$ scheme, the
basic system provides unconditional FindKey-CPA security: after an
execution of the key generation protocol, even an all powerful PKG
does not have any information on the component $d_3$ that is
eventually part of the private key obtained by the new user.
\begin{theorem} \label{findkey}
In the information theoretic sense, no adversary has an advantage in
the $\mathrm{FindKey}\textrm{-}\mathrm{CPA}$ game.
\end{theorem}
\begin{proof}
The proof  directly follows from the perfect  hiding property of
\mbox{Pedersen}'s commitment \cite{Ped91} and the perfect witness
indistinguishability of  the  protocol \cite{Ok92} for proving
knowledge of a discrete logarithm representation. Since the
commitment $R=h^{t_0} \cdot X^{\theta}$ and the proof of knowledge
of $(t_0,\theta)$ perfectly hide $t_0$ to the PKG, all elements of
$\Z_p^*$ are equally likely values of $d_3=t_0+t_1$ as for the last
part of the user's eventual private key. \qed
\end{proof}
The original version of the paper \cite{LV09} describes  a  hybrid
variant of the scheme that provides white-box FindKey-CCA security
using authenticated symmetric encryption in the fashion of
\cite{KD04,SG04,HK07} so as to reject all invalid ciphertexts with
high probability. In this version, we only consider schemes with the
weak black-box traceability property.
\begin{theorem} \label{newkey}
In the selective-ID $\mathrm{ComputeNewKey}$ game, any PPT adversary
has negligible advantage assuming that the Diffie-Hellman assumption
holds.
\end{theorem}
\begin{proof}
For simplicity, we prove the result using an equivalent formulation
of the Diffie-Hellman problem which is to find $h^{1/x}$   given $(g,h,X=g^x)$. \\
\indent At the outset of the game, $\A$ declares the identity
$\ID^\star$ for which she aims at   finding two   private keys
$d_{\ID^\star}^{(1)}$, $d_{\ID^\star}^{(2)}$ comprising distinct
values of $d_3=t$. Then, the simulator $\B$ prepares   the PKG's
public key as follows. Elements $h$ and $X$ are  taken from the
modified Diffie-Hellman instance $(g,h,X)$. As in the proof of
theorem \ref{IND-ID-CPA-proof}, $\B$   defines $Z=g^{-\ID^\star}
\cdot X^{\alpha}$ for a randomly chosen $\alpha \sample \Z_p^*$. To
define $Y$, it chooses   random values $\gamma,t_1' \sample \Z_p^*$
and
sets $Y=X^{\gamma} \cdot h^{-t_1'}$.

\begin{description}
\item[Queries:] in this game, $\A$ is allowed to query  executions of the key
generation protocol w.r.t. any identity, including $\ID^\star$. The only requirement is that queried identities be distinct.
\begin{itemize}
\item[-] For an identity $\ID \neq \ID^\star$, $\B$ can proceed {\it exactly}
as suggested by relation (\ref{PPK}) in the proof of theorem
\ref{IND-ID-CPA-proof} and does not need to rewind $\A$.
\item[-] When $\ID = \ID^\star$, $\B$ conducts the following steps.
When $\A$ supplies a group element $R=h^{t_0} \cdot X^\theta$ along
with a WI proof of knowledge of $(t_0,\theta)$, $\B$ uses the
knowledge extractor of the proof of knowledge that allows extracting
  a representation $(t_0,\theta)$ of $R$ by rewinding $\A$. Next,
  $\B$ computes $t_1=t_1'-t_0$
picks $r \sample \Z_p^*$ and returns
\begin{eqnarray} \label{PPK-ID}
(d_1',d_2',d_3' )=\big( g^{\gamma + \theta} \cdot (g^{\ID^\star}
\cdot Z)^{ r}, ~X^{r},~  t_1    \big).
\end{eqnarray} To see that
the above tuple has the appropriate shape, we note that
\begin{eqnarray*}
(Y \cdot R  \cdot h^{t_1})^{1/x}  = (Y \cdot h^{t_0+t_1} \cdot
X^\theta)^{1/x}   =    (Y \cdot h^{t_1'} \cdot X^\theta )^{1/x}  =
  g^{\gamma +  \theta}.
\end{eqnarray*}
\end{itemize}
\item[Output:] upon its termination, $\A$ is expected to come up with
 distinct valid private keys
$d_{\ID^\star}^{(1)}=(d_1^{(1)},d_2^{(1)},d_3^{(1)})$ and
$d_{\ID^\star}^{(2)}=(d_1^{(2)},d_2^{(2)},d_3^{(2)})$, such that
$t=d_3^{(1)} \neq d_3^{(2)}=t'$, for the identity $\ID^\star$. Given
that we must have
\begin{eqnarray*}
& & d_1^{(1)}= (Y \cdot h^{t})^{1/x} \cdot X^{\alpha r}  \qquad
\quad  ~~
d_2^{(1)}=X^{r}  \\
& & d_1^{(2)}= (Y \cdot h^{t'})^{1/x} \cdot X^{\alpha r'}  \qquad
\quad  d_2^{(2)}=X^{r'}
\end{eqnarray*}
for some values $r,r'   \in \Z_p$, $\B$ can extract  $ h^{1/x} =
\Big( \frac{ d_1^{(1)} / {d_2^{(1)}}^{\alpha}  } {
d_1^{(2)}/{d_2^{(2)}}^\alpha } \Big)^{\frac{1}{t-t'}} .  $
\qed
\end{description}
\end{proof}
  We  note that, in the above proof, the simulator does not have
  to rewind all executions of the key generation protocol but only
  one, when the adversary asks for a private key corresponding to the   target identity
$\ID^\star$ (recall that all queries involve distinct identities).
Given that the number of rewinds is constant, the proof still goes
through when the simulator is presented with many concurrent key
generation queries. If other executions of the protocol (that
necessarily involve identities $\ID\neq \ID^\star$) are nested
within the one
 being rewinded when dealing with   $\ID^\star$, the simulator simply runs them
as an honest verifier would in the proof of knowledge and calculates
the PKG's output as per relation (\ref{PPK}) in the proof of theorem
\ref{IND-ID-CPA-proof}. Thus, the initial rewind does not trigger
any other one and the simulation still takes polynomial time in a
concurrent
setting.

\medskip\noindent \textsc{Adaptive-ID Security.} The scheme can obviously be
made IND-ID-CPA if Waters' ``hash function'' $F(\ID)=
u'\prod_{j=1}^n u_i^{i_j}$ -- where $\ID=i_1\ldots i_n \in
\{0,1\}^n$ and $(u',u_1,\ldots,u_n) \in \G^{n+1}$ is part of
$\mathsf{mpk}$ -- supersedes the Boneh-Boyen identity hashing
$F(\ID)=g^{\ID} \cdot Z$. The number theoretic hash function $F$ is
chosen so as to equal $F(\ID)=g^{J_1(\ID)} \cdot X^{J_2(\ID)}$ for
integer-valued functions $J_1,J_2$ that are computable by the
simulator. The security proof relies on the fact that $J_1$ is small
in absolute value and cancels with non-negligible probability
proportional to
$1/q(n+1)$, where $q$ is the number of key generation queries. \\
\indent When extending the proof of theorem \ref{newkey} to the
adaptive setting, an adversary with advantage $\varepsilon$ allows
solving CDH with probability $\varepsilon /8q^2 (n+1)$. The reason
is that the simulator  has to guess beforehand which key generation
query will involve the target identity $\ID^\star$. If $\ID^\star$
is expected to appear in the $j^{\textrm{th}}$ query, when the
latter is made, $\B$ rewinds $\A$ to extract $(t_0,\theta)$ and uses
the special value $t_1'$ to answer the query as per
$(\ref{PPK-ID})$. With probability $1/q$, $\B$ is fortunate when
choosing $j \sample \{1,\ldots,q\}$ at the beginning and, again,
$J_1(\ID^\star)$ happens to cancel with probability $1 /8q (n+1)$
for the target identity.

\section{Weak Black-Box Traceability} \label{bb-tracing}
\label{weak-BB} Theorem \ref{newkey} showed the infeasibility for
users to compute another key from a different family given their
private key. In these regards, a decryption key implements a
``$1$-copyrighted function'' -- in the terminology of
\cite{NSS99,KY02} -- for the matching identity. Using this property
and the perfect white-box FindKey-CPA security, we describe a
black-box tracing mechanism that protects the user from a dishonest
PKG as long as the latter is
withheld access to a decryption oracle. \\
\indent The tracing strategy is close to the one used by Kiayias and
Yung \cite{KY02} in $2$-user traitor tracing schemes, where the
tracer determines which one out of two subscribers produced a pirate
decoder. In our setting, one rather has to decide whether an
$\varepsilon$-useful decryption device stems from the PKG or the
user himself.
\begin{description}
 \item[Trace$^{\mathbb{D}}(\mathsf{mpk},d_{\ID},\varepsilon)$:] given a well-formed private key $d_{\ID}=({d_1},{d_2},{d_3})$ belonging to a user of identity $\ID$
 and oracle access to a decoder $\mathbb{D}$ that decrypts ciphertexts encrypted for $\ID$ with
 probability $\varepsilon$, conduct the following steps.
 \begin{itemize}
 \item[a.]  Initialize a counter $ctr \leftarrow 0$ and repeat the next steps $L=16\lambda / \varepsilon $ times:
 \\ \vspace{-0.3 cm}
 \begin{itemize}
\item[1.] Choose distinct exponents  $s,s' \sample \Z_p^*$ at random, compute $C_1=X^s$, $C_2=(g^\ID
\cdot Z)^s$ and $C_{3}=e(g,h)^{s'}$.
\item[2.] Calculate $ C_4=m \cdot    e(C_1,d_1) / \big(e(C_2,d_2)  \cdot
C_{3}^{d_3} \big)
   $ for a randomly chosen message $m \in \G_T$.
\item[3.] Feed the decryption device $\mathbb{D}$ with
$(C_1,C_2,C_{3},C_4)$. If $\mathbb{D}$ outputs
$m' \in \G_T$ such that $m' = m$, increment $ctr$. \\
\vspace{-0.3 cm}
 \end{itemize}
 \item[b.] If $ctr=0$, incriminate the PKG. Otherwise,  incriminate the user. \\
 \vspace{-0.3 cm}
 \end{itemize}
\end{description}
The soundness of this algorithm is proved using a similar technique
to \cite{ADMNPS}. To ensure the independence of iterations,  we
assume (as in \cite{ADMNPS}) that pirate devices are stateless, or
resettable, and do not retain information from prior queries: each
decryption query is answered as if it were the first one and, in
particular, the pirate device cannot self-destruct.
\begin{theorem} \label{black-box-newkey}
Under the mDBDH assumption, dishonest users have negligible chance
to produce a decryption device $\mathbb{D}$  that makes the tracing
algorithm incriminate the PKG in the selective-ID ComputeNewKey
game.
\end{theorem}
\begin{proof}
The tracing algorithm points to  the PKG if it ends up with $ctr=0$.
The variable $ctr$ can be seen as the sum of $L=16\lambda /
\varepsilon$ independent random variables $X_i \in \{0,1\}$ having
the same expected value $p_1$. We have $\mu=\mathbf{E}[ctr]=Lp_1$.
The Chernoff bound tells us that, for any real number $\omega$ such
that $0\leq \omega \leq 1$,  $\PR[ ctr < (1 -\omega) \mu ]<
\exp(-\mu \omega^2/2)$. Under the mDBDH assumption, we certainly
have $\mathbf{Adv}^{\textrm{mDBDH}}(\lambda) \leq \varepsilon/2$
(since $\varepsilon/2$ is presumably non-negligible). Lemma
\ref{ctr-plus} shows that   $p_1 \geq \varepsilon -
\mathbf{Adv}^{\textrm{mDBDH}}(\lambda)$, which implies
\begin{eqnarray} \label{chernoff}
 \mu=Lp_1 \geq L(\varepsilon - \mathbf{Adv}^{\textrm{mDBDH}}(\lambda)
) \geq \frac{L \varepsilon}{2} = 8 \lambda  .
\end{eqnarray} With
$\omega=1/2$, the Chernoff bound guarantees  that $$  \PR[ctr < 1 ]
< \PR[ctr < 4 \lambda ] = \PR [ctr < \mu/2]< \exp(-\mu/8)
=\exp(-\lambda).
$$ \qed
\end{proof}
\begin{lemma} \label{ctr-plus}
In the selective-ID ComputeNewKey game, if $\mathbb{D}$ correctly
opens well-formed ciphertexts with probability $\varepsilon$, the
probability that an iteration of the tracing algorithm increases
$ctr$ is at least $p_1 \geq \varepsilon -
\mathbf{Adv}^{\mathrm{mDBDH}}(\lambda)$.
\end{lemma}
\begin{proof}
We consider two games called Game$_0$ and Game$_1$ where the
adversary $\A$ is faced with a ComputeNewKey challenger $\B$ and
produces a decryption device $\mathbb{D}$ which is provided with
ciphertexts during a tracing stage. In Game$_0$, $\mathbb{D}$ is
given a properly formed encryption of some plaintext $m$ whereas it
is given a ciphertext $C$ where $C_{3}$ has been changed in
Game$_1$. In either case, we call $p_i$ (with $i\in \{0,1\}$) the
probability that $\mathbb{D}$ returns the  plaintext $m$ chosen by $\B$. \\
\indent In the beginning of Game$_0$,
 $\A$ chooses a target identity
$\ID^\star$ and  $\B$ defines the system parameters as $X=g^c$,
$h=g^b$, $Y=X^{\gamma} \cdot h^{-t^\star}$ and $Z=g^{-\ID^\star}
\cdot X^\alpha$ for random $\alpha,\gamma,t^\star \sample \Z_p^*$.
Then, $\A$ starts making key generation queries that are treated
using the same technique as in the proof of theorem \ref{newkey}.
Again,  $\B$ only has to  rewind the WI proof when the query
pertains to   $\ID^\star$.  \\ \indent At the end of the game, $\A$
outputs a decryption box $\mathbb{D}$ that correctly decrypts a
fraction $\varepsilon$ of ciphertexts. Then, $\B$ constructs a
ciphertext $C$ as
$$C_1=g^a, \qquad C_2=(g^a)^\alpha , \qquad C_{3 }=T , \qquad C_4=m \cdot   \frac{ e(C_1,d_1) }{
 e(C_2,d_2)  \cdot   C_{3 }^{t^\star} }       $$ where $T \in \G_T$.

In Game$_0$, $\B$ sets $T=e(g,g)^{ab/c}$
so that we have $C_{3 }=e(g,h )^{a/c}$ and $C$ is a valid ciphertext
  (for the encryption exponent $s=a/c$)  that $\mathbb{D}$ correctly decrypts with probability $\varepsilon$. In this case, $\mathbb{D}$
 thus outputs $m'=m \in \G_T$   with probability   $p_0=\varepsilon$.    In Game$_1$,  $T$ is chosen as a random element of
 $\G_T$ and
$C=(C_1,C_2,C_{3 }, C_4)$ has the distribution of a ciphertext
produced by the tracing stage and $\mathbb{D}$ must output a
plaintext $m' = m$ with probability $p_1$. It is clear that
$|p_0-p_1| \leq \mathbf{Adv}^{\textrm{mDBDH}}(\lambda) $ and we thus
have $p_1 \geq \varepsilon -
\mathbf{Adv}^{\textrm{mDBDH}}(\lambda)$. \qed
\end{proof}
\noindent The proofs of theorem \ref{black-box-newkey} and lemma
\ref{ctr-plus}   extend  to the adaptive-ID setting using the same
arguments as in the last paragraph of section \ref{efficient-DBDH}.
     As   mentioned in the remark at the end of section
\ref{sec-first-scheme} in section \ref{efficient-DBDH}, proving
adaptive-ID white-box security against dishonest users incurs a
quadratic degradation factor in the number of adversarial queries.
When transposing the proof of lemma \ref{ctr-plus} to the
adaptive-ID model, we are faced with the same  quadratic degradation
in $q$ and the bound on $p_1$ becomes $p_1 \geq \varepsilon - 8
\cdot q^2 (n+1) \cdot \mathbf{Adv}^{\mathrm{mDBDH}}(\lambda)$. The
proof of theorem \ref{black-box-newkey} goes through as long as
$\varepsilon \geq 16 \cdot q^2 \cdot (n+1) \cdot
\mathbf{Adv}^{\mathrm{mDBDH}}(\lambda)$ (so that $p_1 \geq
\varepsilon/2$). Since $q$ is polynomial, this is asymptotically the
case since $q^2 \cdot (n+1) \cdot
\mathbf{Adv}^{\mathrm{mDBDH}}(\lambda)$ remains
negligible under the mDBDH assumption. \\
\indent The system   turns out to be the first scheme that is
amenable for weak black-box traceability against dishonest users in
the adaptive-ID sense. Due to their reliance on attribute-based
encryption techniques (for which only selective-ID adversaries were
dealt with so far), earlier (weak) black-box  A-IBE proposals
\cite{Goyal,GLSW} are only known to provide selective-ID security
against dishonest users.
 \\
\indent As for the security against dishonest PKGs, we  observed
that, in the FindKey-CPA game, the
 last part $d_3^{(1)}=t$ of the user's private key is perfectly hidden to the malicious PKG after the key generation protocol. Then,
a pirate decoder $\mathbb{D}$   made  by the PKG has negligible
chance of decrypting ciphertexts where $C_3$ is random
  in the same way as the user would. When the user comes
across $\mathbb{D}$ and takes it to the court, the latter runs the
tracing algorithm using $\mathbb{D}$ and the user's well-formed
  key
$d_{\ID}^{(1)}=(d_{1}^{(1)},d_{2}^{(1)},d_{3}^{(1)})$ for which
$d_{3}^{(1)}$ is independent of $\mathbb{D}$.

\begin{lemma} \label{findkey-black-box} In the FindKey-CPA game, one iteration of the tracing
algorithm increases $ctr$ with probability at most $1/p$.
\end{lemma}
\begin{proof}  In an
iteration of the tracing stage, $\mathbb{D}$ is given
$C=(C_1,C_2,C_3,C_4)$ such that $C_1=X^s$, $C_2=(g^{\ID} \cdot
Z)^s$, $C_3=e(g,h)^{s'}$ and   $C_4=m \cdot e(g,Y)^s \cdot
e(g,h)^{(s-s')t}$ for distinct  $s,s' \sample \Z_p^*$. Since
$\mathbb{D}$ has no information on $d_3^{(1)}=t$, for any plaintext
$m \in \G_T$, there is a value $d_3^{(1)}$ that explains $C_4$ and
it comes that $\mathbb{D}$ returns the one chosen by the tracer with
   probability   $1/p$. \qed
\end{proof}
We note that a pirate device $\mathbb{D}$ generated by the dishonest
PKG is able to recognize invalid ciphertexts in the tracing stage
(as it may contain the master secret $x$). However, as long as
$\mathbb{D}$
  is assumed stateless, it cannot shutdown or self-destruct when detecting a tracing attempt.  Moreover, with all but negligible probability,
  it will never
be able to decrypt such invalid ciphertexts in the same way as the
owner of $d_{\ID}^{(1)}$ would.
\begin{theorem}
In the black-box FindKey-CPA game, a dishonest PKG has negligible
advantage.
\end{theorem}
\begin{proof}
The dishonest PKG is not detected if it outputs a decryption box for
which the tracing   ends with a non-zero value of   $ctr$. From
lemma \ref{findkey-black-box}, it easily comes that $\PR[ctr \neq 0]
= \PR[ctr \geq 1] \leq L/p =16 \lambda / (\varepsilon p) \leq 16
\lambda / (2^\lambda \varepsilon) $. \qed
\end{proof}

\indent To secure the scheme against chosen-ciphertext attacks and
preserve the weak black-box property, we can use the
Canetti-Halevi-Katz \cite{CHK} technique or its optimizations
\cite{BK05,BMW} that do not affect the tracing algorithm.

\section{Extension to Gentry's IBE} \label{Gentry}

In this section, we show how to apply the weak black-box tracing
mechanism of
 section \ref{bb-tracing} to   Gentry's IBE. The resulting A-IBE
 system  is obtained by bringing a simple modification to the key
 generation protocol of Goyal's first scheme \cite{Goyal} so as to
 perfectly hide the user's key family from the PKG's view while preserving the efficiency of the whole scheme. \\
 \indent The advantage of this scheme is to directly provide adaptive-ID
 security against dishonest users and under   reductions that are are just as tight as in Gentry's
 system. This benefit comes at the expense of sacrificing the concurrent
security of the key generation protocol (as security proofs require
to rewind  at each key generation query) and relying
 on a  somewhat strong  assumption.
\begin{definition}[\cite{Ge06}] \label{ADBDHE-def}
In bilinear groups $(\G,\G_T)$,  the    {\bf $q$-Decision Augmented
Bilinear Diffie-Hellman Exponent Problem} ($q$-ADBDHE) is to
distinguish the distribution
$\big(g,g^{\alpha},\ldots,g^{(\alpha^q)},h,h^{(\alpha^{q+2})},e(g,h)^{(\alpha^{q+1})}\big)
$ from the distribution
$\big(g,g^{\alpha},\ldots,g^{(\alpha^q)},h,h^{(\alpha^{q+2})},e(g,h)^{\beta}\big)
$, where   $\alpha,\beta \sample \Z_p^*$ and $h\sample \G^* $. The
advantage
$\mathbf{Adv}^{q\textrm{-}\mathrm{ADBDHE}}_{\G,\G_T}(\lambda)$ of a
distinguisher $\B$
   is  defined as in definition \ref{DBDH-def}
\end{definition}

In the description hereafter, the encryption and decryption
algorithms are exactly as in   \cite{Ge06}. Since the key generation
protocol  perfectly conceals the user's key family, we can  apply
the same weak black-box tracing mechanism as in section
\ref{bb-tracing}. The resulting system turns out to be the most
efficient adaptive-ID secure weakly black-box A-IBE to date.

\begin{description}
\item[Setup:]
given a security parameter $\lambda \in \mathbb{N}$, the PKG chooses
bilinear groups $(\G,\G_T)$ of  order $p>2^\lambda$  with a
generator $g \sample \G$. It picks $h,g \sample \G$ and $\alpha
\sample \Z_p^*$ at random. It defines the master key as
$\mathsf{msk}:=\alpha$ and the master public key is defined to be
 $\mathsf{mpk}:=(g,g_1=g^{\alpha},h).$ \\ \vspace{-0.3 cm}
\item[Keygen$^{(\mathrm{PKG},\mathsf{U})~}$:] the user $\U$ and the PKG interact in the following
  protocol. \\
\vspace{-0.3 cm}
\begin{itemize}
\item[1.]  $\U$ picks $t_0, \theta \sample \Z_p^*$ and sends a  commitment $R=g^{-t_0} \cdot (g_1 \cdot g^{-\ID})^\theta$ to the PKG. He
also   gives
  an  interactive   witness
indistinguishable proof of knowledge of the pair  $(t_0,\theta)$. \\
\vspace{-0.3 cm}
\item[2.] The PKG   outputs $\bot$ if the proof of knowledge is invalid. Otherwise, it picks $t_1 \sample \Z_p^*$ and returns
\begin{eqnarray} \label{GE-priv-key-1}
d_{\ID}'=(d',t_{\ID}')=\Big( ( h \cdot R   \cdot
g^{-t_1})^{1/(\alpha -\ID)} ,~t_1 \Big).
\end{eqnarray}
\item[3.]   $\U$  computes
 $d_{\ID}=({d'}/g^{\theta} ,~t_{\ID}'+t_0)$
which should equal
\begin{eqnarray} \label{GE-priv-key-2}
d_{\ID}=({d},t_{\ID})=\Big( ( h \cdot
g^{-(t_0+t_1)})^{1/(\alpha-\ID)} ,~t_0+t_1 \Big).
\end{eqnarray}
  Then, $\U$ checks whether $d_{\ID}$  satisfies the
relation
\begin{eqnarray} \label{GE-valid-priv-1}
e(d ,g_1 \cdot g^{-\ID})&=& e(h,g) \cdot e(g,g)^{-t_{\ID}} .
\end{eqnarray}
If so, he sets his private key as $d_{\ID}$, which   belongs
to the key family   identified by $n_F=t_{\ID}=t_0+t_1$. He outputs $\bot$ otherwise. \\
\vspace{-0.3 cm}
\end{itemize}
\item[Encrypt:] to encrypt   $m \in \G_T$ given $\mathsf{mpk}$ and $\ID$, choose $s \sample \Z_p^*$ and compute \\
\vspace{-0.1 cm}
$$ C= \big( C_1,C_2,C_3  \big)= \Big(  \big(g_1  \cdot g^{- \ID} \big)^s     ,~ e(g,g)^s   ,~ m \cdot e(g,h)^s    \Big). $$
\item[Decrypt:] given $C= \big(C_1,C_2,C_3   \big) $ and $d_{\ID}=({d},t_{\ID})$,
 compute
 \begin{eqnarray*} \label{decrypt}
   m=C_3 \cdot \Big(  e(C_1,d)   \cdot C_2^{t_{\ID}}
   \Big)^{-1}
 \end{eqnarray*}
 \item[Trace$^{\mathbb{D}}(\mathsf{mpk},d_{\ID},\varepsilon)$:] given a valid private key $d_{\ID}=(d,t_{\ID})$
  belonging to   user   $\ID$
 and   a $\varepsilon$-useful  pirate decoder $\mathbb{D}$, conduct the following
 steps. \\ \vspace{-0.3 cm}
 \begin{itemize}
 \item[a.]  Set $ctr \leftarrow 0$ and repeat the next steps $L=16\lambda / \varepsilon $ times:
 \\ \vspace{-0.3 cm}
 \begin{itemize}
\item[1.] Choose    $s,s' \sample \Z_p^*$ such that $s \neq s'$  and  set $C_1=(g_1 \cdot g^{-\ID})^s$  and $C_{2}=e(g,h)^{s'}$.
\item[2.] Compute $ C_3=m \cdot   e(C_1,d)   \cdot C_2^{t_{\ID}}
   $ for a random  message $m \in \G_T$.
\item[3.] Feed the decryption device $\mathbb{D}$ with
$(C_1,C_2,C_{3})$. If $\mathbb{D}$ outputs
$m' \in \G_T$ such that $m' = m$, increment $ctr$. \\
\vspace{-0.3 cm}
 \end{itemize}
 \item[b.] If $ctr=0$, incriminate the PKG. Otherwise,  incriminate the user. \\
 \vspace{-0.3 cm}
 \end{itemize}
\end{description}

The IND-ID-CPA security of the scheme can be simply reduced to that
of Gentry's IBE as shown in the proof of the next theorem.
\begin{theorem} \label{proof-version-Gentry}
Any IND-ID-CPA adversary against the above A-IBE implies an
IND-ID-CPA attacker against Gentry's IBE.
\end{theorem}
\begin{proof}
Let us assume an IND-ID-CPA adversary $\A$ in the game described by
definition \ref{sec-def-A-IBE}. We show that $\A$ gives rise to an
IND-ID-CPA adversary $\B$ against Gentry's IBE. \\
\indent Our adversary $\B$ receives a master public key
$\mathsf{mpk}=(g,g_1,h)$ from her challenger. When  $\A$ makes a key
generation request for an identity $\ID$, $\B$ queries her own
challenger to extract a private key $d_{\ID}=(d,t_{\ID})=\big(
(h\cdot g^{-t_{\ID}})^{1/(\alpha-\ID)},t_{\ID} \big)$ and starts
executing the key generation protocol with in interaction with $\A$.
The latter first supplies a commitment $R=g^{-t_0} \cdot (g_1 \cdot
g^{-\ID})^{\theta}$ and an interactive WI proof of knowledge of the
pair $(t_0,\theta)$. Using the knowledge extractor of the proof of
knowledge, $\B$ extracts $(t_0,\theta)$ by rewinding $\A$   and
returns $d_{\ID}=(d',t_{\ID}')$, where  $t_{\ID}'=t_{\ID}-t_0$ and
$d'=d_{\ID} \cdot g^{\theta}$. \\
\indent In the challenge phase, $\A$ chooses a target identity
$\ID^\star$ and messages $(m_0,m_1)$, which  $\B$ forwards to her
own challenger. The latter provides $\B$ with a challenge ciphertext
$(C_1,C_2,C_3)$ which is relayed to $\A$. After a second series of
key generation queries, $\A$ outputs a bit $d\in \{0,1\}$, which is
also $\B$'s output. It is easy to see that, if $\A$ is successful,
so is $\B$. \qed
\end{proof}
We now turn to prove the weak black-box traceability property.

\begin{lemma}
In the Adaptive-ID ComputeNewKey game and for a $\varepsilon$-useful
device $\mathbb{D}$, the probability that an iteration of the
tracing algorithm increases $ctr$ is at least $p_1 \geq \varepsilon
- \mathbf{Adv}^{q\textrm{-}\mathrm{ADBDHE}}_{\G,\G_T}(\lambda) $,
where $q$ is the number of key generation queries.
\end{lemma}
\begin{proof} The proof is very similar to the proof of IND-ID-CPA
security in \cite{Ge06}. For the sake of contradiction, let us
assume that, in an iteration of the tracing procedure, the
probability $p_1$ that $\mathbb{D}$ returns the message chosen by
the tracer is significantly smaller than $\varepsilon$. Then, we
can construct a distinguisher $\B$ for the $q$-ADBDHE assumption. \\
\indent The distinguisher $\B$ takes as input a tuple
$(g,g^{\alpha},\ldots,g^{(\alpha^q)},h,h^{(\alpha^{q+2})},T)$ and
aims at deciding if $T=e(g,h)^{(\alpha^{q+1})}$. It generates the
master public key in such a way that $h=g^{f(\alpha)}$, for some
random  polynomial $f(X) \in \Z_p[X]$ of degree $q$. At each key
generation query, $\B$ first computes a valid private key
$d_{\ID}=(d,t_{\ID})$ for the identity $\ID$, by setting
$t_{\ID}=f(\ID)$ as in the proof of theorem 1 in \cite{Ge06}. Then,
in the interactive key generation protocol, $\A$ sends a commitment
$R=g^{-t_0} \cdot (g_1 \cdot g^{-\ID})^{\theta}$ and proves
knowledge of the pair $(t_0,\theta)$, which $\B$ extracts by
rewinding $\A$ as in the proof of theorem
\ref{proof-version-Gentry}. As in the latter, $\B$ replies with a
well-distributed pair $d_{\ID}'=(d',t_{\ID}')$, where
$t_{\ID}'=t_{\ID}-t_0$ and $d'=d  \cdot g^{\theta}$. \\
\indent The game ends with $\A$ outputting an identity $\ID^*$, a
private key $d_{\ID^\star}=(d^\star,t_{\ID^\star}^\star)$ and a
$\varepsilon$-useful device. In the tracing stage, $\B$   first
expands the monic polynomial
$F(X)=(X^{q+2}-{\ID^\star}^{q+2})/(X-\ID^\star)=X^{q+1}+F_{q} X^q +
\cdots + F_1 X +F_0$. Then, $\B$ chooses a   plaintext $m\sample
\G_T$ and computes $C=(C_1,C_2,C_3)$ as
\begin{eqnarray*}
C_1=\frac{h^{(\alpha^{q+2})}}{h^{({\ID^\star}^{q+2})}}  \qquad C_2=
T \cdot e\big(h,\prod_{j=0}^q (g^{(\alpha_j)F_j}) \big) \qquad C_3=m
\cdot e(C_1,d^\star) \cdot C_2^{t_{\ID^\star}^\star}.
\end{eqnarray*}
If   $\mathbb{D}$ returns the correct plaintext $m$, the
distinguisher $\B$ halts and return $1$. As in \cite{Ge06},
$(C_1,C_2,C_3)$ is a well-formed ciphertext with the encryption
exponent $s=\log_g(h)F(\alpha)$ if $T=e(g,h)^{(\alpha^{q+1})}$. In
this case, $\B$ returns $1$ with probability $\varepsilon$ since
$\mathbb{D}$ is a $\varepsilon$-useful device. By assumption, the
probability that $\B$ returns $1$ when $T$ is random is
significantly smaller than $\varepsilon$. Therefore, $\B$ has
non-negligible advantage as a distinguisher against the $q$-ADBDHE
assumption. \qed
\end{proof}

\begin{theorem} \label{newkeyGentry}
In the adaptive-ID $\mathrm{ComputeNewKey}$ game, any PPT adversary
has negligible advantage assuming that the ADBDHE assumption holds.
\end{theorem}

\begin{proof}
The proof  is completely analogous to that of theorem
\ref{black-box-newkey}. \qed
\end{proof}
The weak black-box  security against dishonest PKGs follows from the
information theoretic secrecy of the user's private key element
$t_{\ID}$ upon termination of the key generation protocol.

\begin{theorem} \label{findkeyGentry}
In the information theoretic sense, no adversary has an advantage in
the $\mathrm{FindKey}\textrm{-}\mathrm{CPA}$ game.
\end{theorem}

To secure the scheme against chosen-ciphertext attacks, we cannot
use hash proof systems as suggested in \cite{Ge06,KV08}. This
technique  would indeed cause the decryption algorithm to reject all
invalid ciphertexts with high probability, which would not be
compatible with our weak black-box tracing mechanism. \\
\indent Fortunately,  CCA2-security can be acquired by applying the
Canetti-Halevi-Katz transformation to a two-receiver variant of the
Gentry-Waters identity-based broadcast encryption (IBBE) scheme
\cite{GW09}: one of the two receivers' identities is set to be the
verification key of a strongly unforgeable one-time signature and
the matching private key is used to sign the whole ciphertext. \\
\indent Our tracing algorithm can be combined with the latter
approach since, in the Gentry-Waters IBBE \cite{GW09}, private keys
have the same shape as in Gentry's IBE and one of the ciphertext
components lives in the group $\G_T$. As already mentioned, the CHK
technique does not affect  traceability as, upon decryption,
ill-formed ciphertexts only get rejected when the one-time signature
verification fails. The computational/bandwidth cost of the
resulting system exceeds that of the above A-IBE construction only
by a small factor.

\section{Extension to Identity-Based Broadcast Encryption}
\label{IBBE}

As already stressed in \cite{Goyal,GLSW}, reducing the required
amount of trust in PKGs is an equally important problem in IBE
schemes and their extensions such as attributed-based encryption or
identity-based broadcast encryption (IBBE). \\
\indent In this section, we thus show how the underlying idea of
previous schemes can be applied to one of the most efficient IBBE
realizations to date.

\subsection{The Boneh-Hamburg IBBE} \label{BH-IBBE}

An  identity-based broadcast encryption scheme, as formalized in
\cite{AKN07}, can be seen as an IBE where ciphertexts can be
decrypted by more than one receiver. Syntactically, it consists of
four algorithms:
\begin{itemize}
\item \textbf{Setup:} given a security parameter and a bound $N$ on the number of receivers per ciphertext, this algorithm
 outputs a master   key pair
 $(\mathsf{mpk},\mathsf{msk})$.
\item \textbf{KeyGen:} is used by the PKG to derive a private key $K_\ID$ for an identity
$\ID$.
\item \textbf{Encrypt:} takes as input a plaintext $m$, a master public key $\mathsf{mpk}$ and a set $S=\{\ID_1,\ldots,\ID_n\}$ of receivers' identities,
where $n\leq N$. It outputs a ciphertext $C$.
\item \textbf{Decrypt:} takes as input the master public key $\mathsf{mpk}$,  a ciphertext $C$, a set of receivers $S=\{\ID_1,\ldots,\ID_n\}$ and a private key $d_{\ID}$ corresponding
to some identity $\ID \in S$. It outputs a plaintext $m$ or $\bot$.
\end{itemize}

In \cite{BH08}, Boneh and Hamburg showed how to turn the
Boneh-Boyen-Goh hierarchical IBE \cite{BBG05} into an efficient IBBE
system with constant-size ciphertexts  and linear-size private keys
in the bound $N$ on the number of receivers per ciphertext. Their
construction was shown to derive from a more general primitive
termed ``spatial encryption''. Its security (in the selective-ID
sense) was established under the following assumption introduced in
\cite{BBG05}.

\begin{definition}  \label{DBDH-def} Let
$(\G,\G_T)$ be bilinear   groups of   order $p$  and $g \in {\mathbb
G}$. The    {\bf $\ell$-Decision  Bilinear Diffie-Hellman Exponent}
($\ell$-DBDHE) problem is, given
$\big(g,g^{\alpha},g^{(\alpha^2)},\ldots,g^{(\alpha^\ell)},g^{(\alpha^{\ell+2})},\ldots,g^{(\alpha^{2\ell})},h,T
\big)\in \G^{2\ell+1} \times \G_T$ for random   $\alpha \sample
\Z_p^*$ and $h \sample \G$, to decide whether
$T=e(g,h)^{(\alpha^{\ell+1})}$. The advantage
$\mathbf{Adv}^{\mathrm{\ell\textrm{-}DBDHE}}_{\G,\G_T}(\lambda)$ of
a distinguisher $\B$ is defined in the usual way.
\end{definition}
In the following, we use the same  notations as in \cite{BH08} and,
for any vector $\mathbf{a}=(a_0,\ldots,a_N) \in \Z_p^{N+1}$,
$g^{\mathbf{a}}$ stands for the vector $(g^{a_0},\ldots,g^{a_N}) \in
\G^{N+1} $.   The description of the Boneh-Hamburg IBBE scheme is as
follows.

\begin{description}
\item[Setup$(\lambda,N)$:] given a security parameter $\lambda \in
\mathbb{N}$ an the maximal number of receivers $N \in \mathbb{N}$
per ciphertext,  choose bilinear groups $(\G,\G_T)$ of prime order
$p>2^\lambda$ and a generator $g \sample \G$. Choose $z \sample \G$
as well a $(N+1)$-vector $\mathbf{h}=(h_0,h_1,\ldots,h_N) \sample
\G^{N+1}$ of random generators so that $h_i=g^{a_i}$  for
$i=0,\ldots,N$ with a randomly chosen $\mathbf{a}=(a_0,\ldots,a_N)
\sample \Z_p^{N+1}$. Finally, pick $\alpha \sample \Z_p^*$, $g_2
\sample \G$ and compute $g_1=g^{\alpha}$. The master public key is
$\mathsf{mpk}=(g,g_1=g^{\alpha},g_2,z,\mathbf{h}=g^{\mathbf{a}})$
while the master secret key is $\mathsf{msk}=(\mathbf{a},\alpha )$.
\\ \vspace{-0.3 cm}

\item[Keygen$(\mathsf{msk},\ID)$:] to generate a private key for an identity $\ID$, choose a random $r  \sample \Z_p^*$ and compute
\begin{eqnarray*}
K_{\ID} &=& (K_{1},K_{2},T_{0},\ldots,T_{N-1}) \\
&=& \big(g_2^{\alpha} \cdot z^{r},~g^{r},~h_1^{r} \cdot h_0^{-\ID
\cdot r}, ~h_2^{r} \cdot h_1^{-\ID \cdot r},\ldots,~h_N^{r} \cdot
h_{N-1}^{-\ID \cdot r} \big)
\end{eqnarray*}
for which the ``delegation component''  $(T_{0},\ldots,T_{N-1} ) \in
\G^N$ can be expressed as  $g^{r \cdot M_1^{t} \cdot \mathbf{a}}$,
for some matrix
 $M_1 \in \Z_p^{(N+1) \times N}$, which will be defined below. \\
 \vspace{-0.3 cm}
\item[Encrypt$(\mathsf{mpk},S,m)$:] to encrypt $m \in \G_T$
for  the receiver set  $S=\{\ID_1,\ldots,\ID_n\}$, where $n\leq N$, \\
\vspace{-0.3 cm}
\begin{itemize}
\item[1.] Expand the polynomial
\begin{eqnarray} \label{polyn}
P(X)=\prod_{i \in S} (X-\ID_i) = \rho_{n} X^n + \rho_{n-1} X^{n-1} +
\cdots + \rho_1 X + \rho_0.
\end{eqnarray}
\item[2.] Pick $s\sample \Z_p^*$ and compute
\begin{eqnarray*}
C=( C_0,C_1,C_2)=\Bigl(  m \cdot e(g_1,g_2)^s,~g^s,~\big(z \cdot
h_0^{\rho_0} \cdot h_1^{\rho_1} \cdots h_n^{\rho_n}\big)^s \Bigr).
\end{eqnarray*}
\end{itemize}
\item[Decrypt$(\mathsf{mpk},K_{\ID},C,S)$:] parse $S$ as $\{\ID_1,\ldots,\ID_n\}$,  $C$ as
$(C_0,C_1,C_2)$ and $K_{\ID}$ as
$$K_{\ID} = (K_{1},K_{2},T_{0},\ldots,T_{N-1}) \in \G^{N+2}. $$
\begin{itemize}
\item[1.]  Expand the polynomial
$$ P_{\ID}(X)=\prod_{\ID_j \in S \backslash \{\ID  \} } (X-\ID_j) = y_{n-1}^{(\ID)} X^{n-1} + y_{n-2}^{(\ID)} X^{n-2} + \cdots + y_1^{(\ID)} X + y_0^{(\ID)} $$
and use its coefficients to compute
\begin{eqnarray} \label{BH-deleg-1}
(D_{\ID},d_{\ID})&=& \bigl(K_{1} \cdot T_0^{y_0^{(\ID)}} \cdot
T_1^{y_1^{(\ID)}} \cdots T_{n-1}^{y_{n-1}^{(\ID)}},~K_{2 } \bigr)\\
\label{BH-deleg-2} &=& \Bigl( g_2^{\alpha} \cdot \big(z \cdot
h_0^{\rho_0} \cdot h_1^{\rho_1} \cdots h_{n}^{\rho_n}
\big)^{r},~g^{r} \Bigr)
\end{eqnarray}
where $\rho_0,\ldots,\rho_n$ are the coefficients of $P(X)$
(calculated as per (\ref{polyn})). \\ \vspace{-0.3 cm}
\item[2.] Recover the plaintext as
\begin{eqnarray} \label{BH-decrypt-eq-1}
m  =   C_0 \cdot e\big(C_1,   D_{\ID}  \big)^{-1} \cdot e\big(C_2, \
d_{\ID}  \big).
\end{eqnarray}
\end{itemize}
\end{description}
\bigskip
 To see why step 1 of the decryption algorithm works, one observes
that, for any polynomials $(X-\ID)$ and $P_{\ID}(X)= y_{n-1}^{(\ID)}
X^{n-1} + y_{n-2}^{(\ID)} X^{n-2} + \cdots + y_1^{(\ID)} X +
y_0^{(\ID)} $, the coefficients of $P(X)=(X-\ID)P_{\ID}(X)=\rho_n
X^n + \cdots + \rho_1 X + \rho_0$ are given by
\begin{eqnarray*}
\mathbf{\rho}=
\begin{pmatrix}
\rho_0 \\ \rho_1 \\ \rho_2 \\ \vdots \\ \rho_n
\end{pmatrix}= M_1 \cdot \mathbf{y}
=
\begin{pmatrix} -\ID & & & & & \\
1 & -\ID &   & & & \\
&  1 & -\ID &   & & \\
& &  \ddots & \ddots  &    & \\
& &   &    & 1   & -\ID   \\
& &  &  &       & 1
\end{pmatrix} \cdot
\begin{pmatrix}
y_0^{(\ID)} \\ y_1^{(\ID)} \\   \vdots \\ y_{n-1}^{(\ID)}
\end{pmatrix},
\end{eqnarray*}
where $M_1 \in \Z_p^{(n+1) \times n}$. Since the latter matrix is
such that
\begin{eqnarray*}
M_1^t \cdot \mathbf{a}|_{n+1} = M_1^t  \cdot
\begin{pmatrix}
a_0 \\ a_1 \\   \vdots \\ a_n
\end{pmatrix}=\begin{pmatrix}
a_1-\ID \cdot a_0 \\ a_2 -\ID \cdot a_1 \\   \vdots \\ a_n-\ID \cdot
a_{n-1}
\end{pmatrix},
\end{eqnarray*}
 for each     private key $K_{\ID}$, the first $n$ delegation
components satisfy
\begin{eqnarray*}
(T_{0},\ldots,T_{n-1})&=& \big(h_1^{r} \cdot h_0^{-\ID \cdot r},
~h_2^{r} \cdot h_1^{-\ID \cdot r},\ldots,~h_n^{r} \cdot
h_{n-1}^{-\ID  \cdot r} \big)    =  g^{r M_1^t \cdot \mathbf{a}}.
\end{eqnarray*}
Therefore, since $\rho  =   M_1 \cdot \mathbf{y} $, we have
\begin{eqnarray*}
(z \cdot \prod_{k=0}^n h_k^{\rho_k})^{r } &=&  z^{r} \cdot g^{r
\cdot \mathbf{\rho}^t \cdot   \mathbf{a}}
  = z^{r} \cdot  g^{r \mathbf{y}^t \cdot M_1^t \cdot \mathbf{a}}
   =  z^{r} \cdot  T_0^{y_{0}^{(\ID)}} \cdots T_{n-1}^{y_{n-1}^{(\ID)}}
\end{eqnarray*}
which explains the transition between relations (\ref{BH-deleg-1})
and (\ref{BH-deleg-2}). To explain the second step of the decryption
algorithm, we note that, for each $\ID \in S$, the pair
$(D_{\ID},d_{\ID})$ satisfies
\begin{eqnarray} \label{BH-rel-priv}
e(D_{\ID},g)=e(g_1,g_2)  \cdot e(z \cdot h_0^{\rho_0} \cdot
h_1^{\rho_1} \cdots h_n^{\rho_n} ,d_{\ID})
\end{eqnarray}
By raising both members of (\ref{BH-rel-priv}) to the power $s   \in
\Z_p^*$, where $s$ is the random encryption exponent, we see why $m$
can be recovered as per
(\ref{BH-decrypt-eq-1}).\\
\indent The security of this scheme was proved \cite{BH08} under the
$(N+1)$-DBDHE assumption  in the selective-ID model. In the context
of IBBE schemes, the IND-sID-CPA model was formalized in
\cite{AKN07}. It requires the adversary to choose upfront ({\it
i.e.}, before seeing $\mathsf{mpk}$) the set
$S^\star=\{\ID_1^\star,\ldots,\ID_{n^\star}^\star\}$ of identities
under which the challenge ciphertext $C^\star$ will be generated.
The adversary is then allowed to query private keys for identities
$\ID_i \not\in S^\star$ and eventually aims at guessing which one
out of two messages of her choice was encrypted in the generation of
$C^\star$.

\subsection{A weak Black-Box Accountable Authority IBBE}

The idea of the scheme in section \ref{efficient-DBDH}  applies to
construct an IBBE scheme with short ciphertexts and accountable
authorities. The syntax of accountable authority IBBE (A-IBBE)
schemes extends that of IBBE systems in the same way as the A-IBE
primitive extends IBE. The resulting  construction goes as follows.

\begin{description}
\item[Setup$(\lambda,N)$:]  is as in the Boneh-Hamburg IBBE but the algorithm chooses an additional random group element $g_3$.
 The master public key thus consists of
$\mathsf{mpk}=(g,g_1=g^{\alpha},g_2,g_3,z,\mathbf{h}=g^{\mathbf{a}})$
while the master secret   is $\mathsf{msk}=(\mathbf{a},\alpha )$.
\\ \vspace{-0.3 cm}

\item[Keygen$^{(\mathrm{PKG},\mathsf{U})~}$:]  the two parties conduct the following interactive steps.
\\
\vspace{-0.3 cm}
\begin{itemize}
\item[1.]  $\U$ picks $t_0, \theta \sample \Z_p^*$ and sends a  commitment $R=g_2^{t_0} \cdot g^\theta$ to the PKG and provides
  an  interactive   WI  proof of knowledge of    $(t_0,\theta)$. \\
\vspace{-0.3 cm}
\item[2.] The PKG   outputs $\bot$ if the proof of knowledge is invalid. Otherwise, it picks $r,t_1 \sample \Z_p^*$ and returns
\begin{eqnarray*} \label{IBBE-priv-key-1}
K_{\ID}' &=& (K_{1}',K_{2}',T_{0}',\ldots,T_{N-1}',t_{\ID}') \\
&=& \big((g_2^{t_1} \cdot R \cdot g_3)^{\alpha} \cdot
z^{r},~g^{r},~h_1^{r} \cdot h_0^{-\ID \cdot r}, ~h_2^{r} \cdot
h_1^{-\ID \cdot r},\ldots,~h_N^{r} \cdot h_{N-1}^{-\ID \cdot r},~t_1
\big)
\end{eqnarray*}

\item[3.]   $\U$ picks $r' \sample \Z_p^*$ and computes
 $K_{\ID}=(K_{1} ,K_{2},T_{0},\ldots,T_{N-1},t_{\ID})$, where   $K_1=(K_1' /g_1^{\theta}) \cdot z^{r'}$, $K_2=K_2' \cdot
 g^{r'}$,
 $T_i=T_i' \cdot (h_{i+1} \cdot h_i^{-\ID})^{r'}$
 for indices $i=0,\ldots,N-1$   and
 $t_{\ID}'+t_0$, so that
\begin{eqnarray*} \label{IBBE-priv-key-2}
K_{\ID}&=& (K_{1} ,K_{2},T_{0},\ldots,T_{N-1},t_{\ID})\\ &=&\Big( (
g_2^{t_0+t_1} \cdot g_3 )^{\alpha} \cdot z^{r''}  ,
~g^{r''},~h_1^{r''} \cdot h_0^{-\ID \cdot r''},  \ldots,~h_N^{r''}
\cdot h_{N-1}^{-\ID \cdot r''}    ,~t_0+t_1 \Big),
\end{eqnarray*}
 where $r''=r+r'$.  Then, $\U$ checks whether $d_{\ID}$  satisfies the
relation
\begin{eqnarray*} \label{IBBE-valid-priv-1}
e(K_1,g)  =   e(g_1,g_2)^{t_{\ID}} \cdot e(g_1,g_3) \cdot e(z,K_2),
\end{eqnarray*}
and  $e(g,T_i)  =  e(K_2,h_{i+1} \cdot h_i^{-\ID} )$  for  each  $i
\in \{0,\ldots,N-1\}$. \\ \vspace{-0.2 cm}
\end{itemize}
\item[Encrypt$(\mathsf{mpk},S,m)$:] to encrypt $m \in \G_T$
for  the receiver set   $S=\{\ID_1,\ldots,\ID_n\}$, where $n \leq N$,  \\
\vspace{-0.2 cm}
\begin{itemize}
\item[1.] Expand $P(X) \in \Z_p[X]$ as
\begin{eqnarray*}
P(X)=\prod_{i \in S} (X-\ID_i) = \rho_{n} X^n + \rho_{n-1} X^{n-1} +
\cdots + \rho_1 X + \rho_0.
\end{eqnarray*}
\item[2.] Choose $s\sample \Z_p^*$ and compute
\begin{eqnarray*}
C &=& ( C_0,C_1,C_2,C_3) \\ &=& \Bigl(  m \cdot
e(g_1,g_3)^s,~g^s,~\big(z \cdot h_0^{\rho_0} \cdot h_1^{\rho_1}
\cdots h_n^{\rho_n}\big)^s,~e(g_1,g_2)^s \Bigr). \\
\end{eqnarray*}
\end{itemize}
\item[Decrypt$(\mathsf{mpk},K_{\ID},C,S)$:] parse $C$ as
$( C_0,C_1,C_2,C_3)$ and $K_{\ID}$ as
$$K_{\ID} = (K_{1},K_{2},T_{0},\ldots,T_{N-1},t_{\ID}) \in \G^{N+2} \times \Z_p. $$
\begin{itemize}
\item[1.] Expand $P_{\ID}(X) \in \Z_p[X]$ as
$$ P_{\ID}(X)=\prod_{\ID_j \in S \backslash \{\ID  \} } (X-\ID_j) = y_{n-1}^{(\ID)} X^{n-1} + y_{n-2}^{(\ID)} X^{n-2} + \cdots + y_1^{(\ID)} X + y_0^{(\ID)} $$
and  compute the decryption key
\begin{eqnarray*}
(D_{\ID},d_{\ID},t_{\ID})&=& \bigl(K_{1} \cdot T_0^{y_0^{(\ID)}}
\cdot
T_1^{y_1^{(\ID)}} \cdots T_{n-1}^{y_{n-1}^{(\ID)}},~K_{2 },~t_{\ID} \bigr)\\
\label{deleg-2} &=& \Bigl( (g_2^{t_{\ID}} \cdot g_3)^{\alpha} \cdot
\big(z \cdot h_0^{\rho_0} \cdot h_1^{\rho_1} \cdots h_{n}^{\rho_n}
\big)^{r},~g^{r},~t_{\ID} \Bigr).
\end{eqnarray*}
\item[2.] Recover the plaintext as
\begin{eqnarray*} \label{IBBE-decrypt-eq-1}
m  =   C_0 \cdot e\big(C_1,   D_{\ID}  \big)^{-1} \cdot e\big(C_2, \
d_{\ID}  \big) \cdot C_3^{t_{\ID}}.
\end{eqnarray*}
\end{itemize}
\item[Trace$^{\mathbb{D}}(\mathsf{mpk},K_{\ID},\varepsilon)$:] given a valid private key $K_{\ID}$
 for the identity  $\ID$
 and   a $\varepsilon$-useful    decoder $\mathbb{D}$, the tracing
 algorithm proceeds in a similar fashion to previous schemes, by
 feeding $\mathbb{D}$ with ciphertexts $C=(C_0,C_1,C_2,C_3)$ and the receiver set $S$. In the generation of
 $C$, $C_1$ and $C_2$ are calculated as   specified by the
 encryption algorithm. On the other hand,
 $C_3$ is chosen as  a random element of $\G_T$ and $C_0$ is obtained by applying the decryption algorithm to $S$ and $(C_1,C_2,C_3)$.
\end{description}
Correctness is implied the fact that the decryption key
$(D_{\ID},d_{\ID},t_{\ID})$ satisfies the relation
$e(D_{\ID},g)=e(g_1,g_2)^{t_{\ID}} \cdot e(g_1,g_3) \cdot e(z \cdot
\prod_{i=0}^{n} h_i^{\rho_i} ,d_{\ID})$ and raising both members to
the power $s$ as in previous schemes. \\
\indent To avoid repeating the work of Boneh and Hamburg, we prove
the security properties of the above A-IBBE system by reducing them
to the IND-sID-CPA security of the underlying IBBE.

\begin{theorem} \label{proof-version-IBBE}
The A-IBBE scheme is secure under the $(N+1)$-DBDHE assumption. More
precisely, any IND-sID-CPA adversary against it implies an equally
successful IND-sID-CPA attacker against the Boneh-Hamburg IBBE.
\end{theorem}
\begin{proof}
We show that an IND-sID-CPA adversary $\A$ against the A-IBBE scheme
   gives rise to a ``real-or-random'' IND-sID-CPA adversary
$\B$ (\emph{i.e.}, in which the adversary $\A$ outputs a single
message $m$ and has to decide whether the challenge ciphertext
$C^\star$ encrypts $m$ or a random
message) against the Boneh-Hamburg  IBBE. Hence, the security of the latter implies the security of our scheme. \\
\indent When $\A$ chooses her set of target identities $S^\star =
\{\ID_1^\star,\dots,\ID_{n^\star}^\star \}$, with $n^\star \leq N$,
our adversary $\B$ forwards $S^\star$ to her own challenger and
receives a master public key
 $\mathsf{mpk}^{\mathrm{BH}}=(g,g_1=g^{\alpha},g_2,z,\mathbf{h}=g^{\mathbf{a}}).$
Then,  $\B$ picks   $t^*,\beta \sample \Z_p^*$, computes $g_3 =
g_2^{-t^*} g^{\beta}$ and provides $\A$ with $\mathsf{mpk}=(g,g_1
,g_2,g_3,z,\mathbf{h})$. \\
\indent At any time, $\A$ may request an execution of the key
generation protocol for an arbitrary identity $\ID \not \in
S^\star$. At the beginning of each such protocol, $\A$ sends a
commitment $R=g_2^{t_0} \cdot g^{\theta}$ and interactively proves
knowledge of $(t_0,\theta)$, which $\B$ extracts by rewinding $\A$.
Then, $\B$ chooses $t_1 \sample \Z_p^*$, sets $t=t_0+t_1$ and
queries her own IND-sID-CPA challenger   to obtain a private key
$$\tilde{K}_{\ID} =
(\tilde{K}_{1},\tilde{K}_{2},\tilde{T}_{0},\ldots,\tilde{T}_{N-1})=\big(g_2^{\alpha}
\cdot z^{r},~g^{r},~h_1^{r} \cdot h_0^{-\ID \cdot r}, ~h_2^{r} \cdot
h_1^{-\ID \cdot r},\ldots,~h_N^{r} \cdot h_{N-1}^{-\ID \cdot r}
\big)
$$ for the  identity $\ID$ chosen by $\A$. The latter is turned
into an A-IBBE private key and re-randomized  by setting
\begin{multline*}
 K_{\ID}  =
(K_{1},K_{2},T_{0},\ldots,T_{N-1})   =  \big(g_1^{\beta} \cdot
\tilde{K}_{1}^{(t-t^\star)} \cdot z^{r'} ,\\
\tilde{K}_{2}^{(t-t^\star)} \cdot g^{r'},
\tilde{T}_{0}^{(t-t^\star)} \cdot (h_1 \cdot h_0^{-\ID})^{r'},
\ldots,\tilde{T}_{n-1}^{(t-t^\star)}  \cdot (h_N \cdot
h_{N-1}^{-\ID})^{r'}  \big),
\end{multline*}
where  $r'\sample \Z_p^*$. The new key  $K_{\ID}$ is easily seen to
have the same distribution as those obtained in step 3 of the key
generation protocol. Finally, $\A$ obtains  the ``blinded key''
$K_{\ID}'=(K_{1}',K_{2}',T_{0}',\ldots,T_{N-1}')$, where $K_1'=K_1
\cdot g_1^{\theta}$. \\
\indent In the challenge phase, $\A$ chooses a pair of target
messages $(m_0,m_1)$. The adversary $\B$  chooses a random plaintext
$m^\star \sample \G_T$, which she sends to her own
``real-or-random'' challenger. The latter replies with a challenge
ciphertext
\begin{eqnarray*}
C^\star=( C_0,C_1,C_2)=\Bigl(  m  \cdot
e(g_1,g_2)^{s^\star},~g^{s^\star},~\big(z \cdot h_0^{\rho_0} \cdot
h_1^{\rho_1} \cdots h_{n^\star}^{\rho_{n^\star}}\big)^{s^\star}
\Bigr).
\end{eqnarray*}
for the receiver set $S^\star = \{\ID_1^*,\dots,\ID_{n^\star}^* \}$,
where $m $ is either $m^\star$ or  a random element of $\G_T$. The
adversary $\B$ picks a random   bit $d \sample \{0,1\}$ and computes
$C' = ( C_0',C_1,C_2,C_0/m^\star)$ where $C_0' = m_d \cdot
(C_0/m^\star)^{-t^*} \cdot e(g_1,C_1)^\beta$ and $C'$ is relayed to
$\A$ as a challenge ciphertext. After a second series of
 key generation queries, $\A$ outputs a bit $d'\in \{0,1\}$, and $\B$
outputs ``real'' if $d'=d$ and ``random'' otherwise. It is easy to
see that, if $C^\star$   encrypts   a random plaintext, then $
C_0/m^\star $ can be expressed as $ C_0/m^\star
=e(g_1,g_2)^{s^\star-s'}$, where $s^\star=\log_g(C_1)$ and for some
$s'\neq 0$. In this case, we obtain that $C_0'=m_d \cdot
e(g_1,g_3)^{s^\star} \cdot e(g_1,g_2)^{s't^\star}$ statistically
hides $m_d$ (and thus $\PR[d'=d]=1/2$) since $\A$ has no information
on $t^*$. In contrast, if $C^\star$ encrypts $m^\star$, then $C'$ is
a valid encryption of $m_d$ for the A-IBBE scheme, so that
$\PR[d'=d]=1/2+\mathbf{Adv}^{\mathrm{BH}\textrm{-}\mathrm{IND}\textrm{-}\mathrm{sID}\textrm{-}\mathrm{CPA}}_{\G,\G_T}(\lambda)$,
where the latter advantage function denotes the maximal
``real-or-random'' advantage of any IND-sID-CPA adversary against
the Boneh-Hamburg IBBE. It comes that $\B$'s advantage in the
real-or-random game is exactly
$\mathbf{Adv}^{\mathrm{BH}\textrm{-}\mathrm{IND}\textrm{-}\mathrm{sID}\textrm{-}\mathrm{CPA}}_{\G,\G_T}(\lambda)$.
\qed
\end{proof}

\begin{lemma} \label{IBEE-newkey}
In the selective-ID ComputeNewKey game and for a
$\varepsilon$-useful decryption device $\mathbb{D}$, the probability
that an iteration of the tracing procedure increases $ctr$ is at
least $p_1 \geq \varepsilon -
\mathbf{Adv}^{(N+1)\textrm{-}\mathrm{DBDHE}}_{\G,\G_T}(\lambda)
 $.
\end{lemma}
\begin{proof}  Let us assume that, at the end of the selective-ID ComputeNewKey game, the dishonest user $\A$ outputs a device $\mathbb{D}$
for which a given iteration of the tracing procedure increments
$ctr$ with  a probability $p_1$, which is significantly smaller than
$\varepsilon$. Then, we show how to obtain an IND-sID-CPA adversary $\B$ against the Boneh-Hamburg IBBE. \\
\indent The adversary $\B$ plays the IND-sID-CPA game against a
challenger $\mathcal{C}^{\mathrm{BH}}$ and plays $\A$'s challenger
in the selective-ID ComputeNewKey game. At the outset of the latter,
$\A$ chooses a target identity $\ID^*$ and $\B$ chooses her set of
target identities as  $S^\star=\{\ID^\star \}$. When seeing the
description of $S^\star$, the IBBE challenger
$\mathcal{C}^{\mathrm{BH}}$ generates a master public key
$\mathsf{mpk}^{\mathrm{BH}}=(g,g_1,g_2,z,\mathbf{h})$. Then, $\B$
chooses $t^\star,\beta \sample \Z_p^*$ and sets $g_3=g_2^{-t^\star}
\cdot g^{\beta}$. The master public key of the A-IBBE system is
defined as $\mathsf{mpk}=(g,g_1,g_2,g_3,z,\mathbf{h} )$ and given to
$\A$.  \\
\indent Then, $\A$ starts making a number of key generation queries.
For each key generation query involving an identity $\ID \neq
\ID^\star $, $\B$ proceeds by invoking her own challenger
$\mathcal{C}^{\mathrm{BH}}$, exactly as in the proof of theorem
\ref{proof-version-IBBE}.
   When $\A$ queries a private key $K_{\ID^\star}$ for the
target identity $\ID^\star$, $\B$ first rewinds the proof of
knowledge so as to extract the pair $(t_0,\theta)$ such that
$R=g_2^{t_0} \cdot g^{\theta}$ in the commitment. Then, it sets
$t_1=t^\star-t_0$ (in such a way that $t=t_0+t_1=t^\star$). In this
case, $\B$ can compute an A-IBBE private key $K_{\ID^\star}$ on her
own (without having to query  $\mathcal{C}^{\mathrm{BH}}$) as
$$(K_1,K_2,T_0,\ldots,T_{N-1},t_{\ID^\star})=\big( g_1^{\beta} \cdot z^r, ~g^r, ~(h_1 \cdot h_0^{-\ID^\star})^r, \ldots,~(h_{N } \cdot h_{N-1}^{-\ID^\star})^r,~t^\star \big), $$
which is well-formed since $g_2^{t^\star} \cdot g_3=g^{\beta}$.
Finally, $\B$ returns  the ``blinded key'' $K_{\ID^\star}'=\big(
g_1^{\theta} \cdot K_1, K_2, T_0, \ldots,T_{N-1},t_1 \big)  $  to
$\A$.\\
 \indent At the end of the game, $\A$ outputs a private key $K_{\ID^\star}$
and a $\varepsilon$-useful device for the identity $\ID^\star$. In
the tracing stage, $\B$ sends a random plaintext $m^\star \sample
\G_T$ to $\mathcal{C}^{\mathrm{BH}}$ who replies with a challenge $(
C_0^\star,C_1^\star,C_2^\star)$, where   $C_0^\star=m^\star \cdot
e(g_1,g_2)^{s^\star}$ and $C_1=g^{s^\star}$ if
$\mathcal{C}^{\mathrm{BH}}$ is playing the ``real'' game. On the
other  hand, if $\mathcal{C}^{\mathrm{BH}}$ decides to play the
``random'' game, $C_0^\star$ is random in $\G_T$. To construct a
ciphertext for the A-IBBE scheme, $\B$ sets $C_3=C_0^\star/m^\star$
(which equals $e(g_1,g_2)^{s^\star}$ in the ``real'' game and
$e(g_1,g_2)^{s'}$, with $s'\neq s^\star$ in the ``random'' game),
$C_1=C_1^\star$ and $C_2=C_2^\star$. To compute $C_0$, $\B$ chooses
$m \sample \G_T$ and calculates
\begin{eqnarray} \label{C0-eq}
C_0  =   m  \cdot e\big(C_1,   D_{\ID^\star}  \big)  \cdot
e\big(C_2, \ d_{\ID^\star}  \big)^{-1} \cdot C_3^{-t^\star},
\end{eqnarray}
where $(D_{\ID^\star},d_{\ID^\star},t^\star)$ is the decryption key
for the identity $\ID^\star$ and the receiver set $S^\star$, which
is obtained from $K_{\ID^\star}$. \\
\indent If   $\mathbb{D}$ returns the correct plaintext $m$, the
distinguisher $\B$ halts and outputs ``real'' (meaning that
$\mathcal{C}^{\mathrm{BH}}$ is playing the ``real'' game).
Otherwise, $\B$ outputs ``random''. In the former  case,
$(C_0,C_1,C_2,C_3)$ is a valid ciphertext for the receiver set
$S^\star =\{ \ID^\star \}$ and $\B$ returns $1$ with probability
$\varepsilon $ since $\mathbb{D}$ is a $\varepsilon$-useful device.
If $\mathcal{C}^{\mathrm{BH}}$ plays the random game, $\log_g(C_1)
\neq \log_{e(g_1,g_2)}(C_3)$ and $(C_0,C_1,C_2,C_3)$ has the
distribution of a ciphertext generated in iterations of the tracing
stage.  In this case, the probability that $\mathbb{D}$ returns the
plaintext $m$ is $p_1 $. By the definition of IND-sID-CPA security
of the IBBE scheme, we must have $\varepsilon-p_1  \leq
\mathbf{Adv}^{\mathrm{BH}\textrm{-}\mathrm{IND}\textrm{-}\mathrm{sID}\textrm{-}\mathrm{CPA}}_{\G,\G_T}(\lambda)
$. Since the result of \cite{BH08} implies that
$\mathbf{Adv}^{\mathrm{BH}\textrm{-}\mathrm{IND}\textrm{-}\mathrm{sID}\textrm{-}\mathrm{CPA}}_{\G,\G_T}(\lambda)
\leq \mathbf{Adv}^{(N+1)\textrm{-}\mathrm{DBDHE}}_{\G,\G_T}(\lambda)
$, the claimed result   follows.  \qed
\end{proof}

\begin{theorem} \label{newkeyIBBE}
In the selective-ID $\mathrm{ComputeNewKey}$ game, any PPT adversary
has negligible advantage assuming that the $(N+1)$-DBDHE assumption
holds.
\end{theorem}
\begin{proof}
Again, the proof is  similar to the one of theorem
\ref{black-box-newkey} and is  omitted. \qed
\end{proof}
As in previous schemes, as long as pirate devices are stateless, no
dishonest PKG can create one that gets the tracing procedure to
accuse the user and the result holds unconditionally.
\begin{theorem} \label{findkeyIBBE}
In the information theoretic sense, no adversary has an advantage in
the $\mathrm{FindKey}\textrm{-}\mathrm{CPA}$ game.
\end{theorem}

We remark that it is possible to re-write the description of our
scheme of section  \ref{efficient-DBDH} in such a way that its
security properties can be reduced to the security of the first
Boneh-Boyen IBE \cite{BB2} (in the same way as we reduced the
security of our A-IBBE to the security of the underlying IBBE).
However, giving a proof from scratch allowed us to avoid rewinding
  as much as possible in section  \ref{efficient-DBDH}. It would be
  interesting to see if, in our A-IBBE, the  number of rewinds can
  also be minimized by giving direct proofs under the $(N+1)$-DBDHE
  assumption for theorem \ref{proof-version-IBBE} and lemma
  \ref{IBEE-newkey}.
\\
\indent It is also noteworthy  that other IBE-related primitives can
be made accountable using the same technique. Due to their algebraic
  similarities with the ``commutative blinding''
IBE family, the ``large-universe'' attribute-based encryption
schemes described in \cite{SW05,GPSW} can easily be tweaked to
support accountability  in the weak black-box  model.

\section{Conclusion}
We described the first A-IBE system allowing for weak black-box
traceability while retaining short ciphertexts and private keys. We
also suggested a white-box variant that dwells secure against
dishonest PKGs equipped with a decryption oracle. In the black-box
setting, it remains an open problem to  achieve the latter property
without significantly degrading the efficiency. \\
\indent In the setting of hierarchical IBE schemes, it would also be
desirable to see how the problem can be addressed. When a pirate
decoder is found to decrypt ciphertexts intended for  a node, one
should be able to determine which ancestor(s) of that node should be
blamed.


\appendix

\section{A Variant with White-Box FindKey-CCA security} \label{cca}
To achieve IND-sID-CCA2 security, we can hybridize the scheme using
an authenticated symmetric encryption
 scheme (as defined in appendix \ref{auth-sym}) as previously
considered    in \cite{SC07,KV08} in the context of identity-based
encryption. The obtained  variant is
  reminiscent of a   version of Gentry's IBE
described in \cite{KV08} and can be proved IND-sID-CCA2 secure in a
completely analogous way.

\begin{description}
\item[Setup:]
is the same as in section \ref{efficient-DBDH} except that the PKG
now chooses two elements $Y_A,Y_B \sample \G$ instead of a single
one $Y$. An authenticated symmetric encryption scheme
$(\mathsf{E},\mathsf{D})$ of keylength $\ell \in \mathbb{N}$, a
secure key derivation function $KDF :\G_T \rightarrow \{0,1\}^\ell$
and  a  target collision-resistant hash function $H  :\{0,1\}^*
\rightarrow \Z_p^*$  are also needed. The master key is set as
$\mathsf{msk}:=x$ and the global
public key is   $\mathsf{mpk}:=(X=g^x,h,Y_A,Y_B,Z,H,KDF,(\mathsf{E},\mathsf{D}))$.    \\
\vspace{-0.3 cm}.
\item[Keygen$^{(\mathrm{PKG},\mathsf{U})~}$:] to obtain a private key for his identity $\ID$, a user $\U$ interacts with the PKG as follows. \\
\vspace{-0.3 cm}
\begin{itemize}
\item[1.]   $\U$ sends $R=h^{t_0} \cdot X^\theta$ to the PKG and proves his knowledge of the underlying pair $(t_0,\theta) \sample (\Z_p^*)^2$ in
a witness indistinguishable fashion.
\item[2.] The PKG   outputs $\bot$ if the proof is incorrect. Otherwise, it picks random values $r_{A}',t_{A,1},r_B',t_B \sample \Z_p^*$ and returns
\begin{eqnarray*} \label{priv-key-1-cca}
d_{\ID,A}' &=& (d_{A,1}',d_{A,2}',d_{A,3}')  =  \Big( ( Y \cdot R
\cdot h^{t_{A,1}})^{1/x} \cdot (g^\ID \cdot
Z)^{r_{A}'},~X^{r_{A}'},~t_{A,1} \Big) \\
d_{\ID,B}' &=& (d_{B,1}',d_{B,2}',d_{B,3}')  =  \Big( ( Y_B  \cdot
h^{t_{B}})^{1/x} \cdot (g^\ID \cdot Z)^{r_{B}'},~X^{r_{B}'},~t_{B}
\Big)
\end{eqnarray*}
\item[3.]   $\U$   computes
 $d_{\ID,A}=({d_{A,1}'}/{g^\theta} \cdot (g^{\ID} \cdot Z)^{r_A''},~d_{A,2}' \cdot X^{r_A''},~{d_{A,3}'}+t_0)$
as well as $d_{\ID,B}=({d_{B,1}'} \cdot (g^{\ID} \cdot
Z)^{r_B''},~d_{B,2}' \cdot X^{r_B''},~{d_{B,3}})$, for randomly
chosen $r_A'',r_B'' \sample \Z_p^*$ so that
\begin{eqnarray*} \label{priv-key-2-cca}
d_{\ID,A}&=&({d_{A,1}},{d_{A,2}},{d_{A,3}})=\Big( ( Y_A \cdot
h^{t_A})^{1/x} \cdot (g^\ID \cdot Z)^{r_A},~X^{r_A},~t_A  \Big) \\
d_{\ID,B}&=&({d_{B,1}},{d_{B,2}},{d_{B,3}})=\Big( ( Y_B \cdot
h^{t_B})^{1/x} \cdot (g^\ID \cdot Z)^{r_B},~X^{r_B},~t_B  \Big)
\end{eqnarray*}
where $t_A=t_0+t_{A,1}$, $r_A=r_A'+r_A''$ and $r_B=r_B'+r_B''$.  He
checks whether $d_{\ID,A}$ and $d_{\ID,B}$ respectively satisfy
\begin{eqnarray} \label{valid-priv-1-cca}
e(d_{A,1},X)&=& e(Y_A,g) \cdot e(h,g)^{d_{A,3}} \cdot e(g^\ID \cdot
Z,d_{A,2}) \\   \label{valid-priv-2-cca} e(d_{B,1},X)&=& e(Y_B,g)
\cdot e(h,g)^{d_{B,3}} \cdot e(g^\ID \cdot Z,d_{B,2}).
\end{eqnarray}
If so, he sets his private key as $(d_{\ID,A},d_{\ID,B}) $ and the
latter belongs
to the family of decryption key identified by $n_F=d_{A,3}=t_A$.   \\
\vspace{-0.3 cm}
\end{itemize}
\item[Encrypt:] to encrypt   $m  $ given $\mathsf{mpk}$ and $\ID$, choose $s \sample \Z_p^*$ and compute \\
\vspace{-0.2 cm}
$$ C= \big( C_1,C_2,C_3,C_4  \big)= \Big(  X^{s}    ,~ (g^\ID \cdot Z)^{s},~e(g,h)^s   ,~ \mathsf{E}_K(m)    \Big)  $$
where $K=KDF(   e(g,Y_A)^s  \cdot e(g,Y_B)^{\kappa s})$ and
$\kappa=H(C_1,C_2,C_3)$. \\ \vspace{-0.3 cm}
\item[Decrypt:] given $C= \big(C_1,C_2,C_3,C_4  \big) $ and $d_{\ID}=(d_{\ID,A},d_{\ID,B})$,
 compute the plaintext $m =\mathsf{D}_{K}(C_4)$ (which may just be $\bot$ if $C_4$ is not a valid authenticated encryption)
 using the key
 \begin{eqnarray} \label{decrypt-cca}
   K= KDF \Big( \frac{e(C_1, d_{A,1} \cdot d_{B,1}^\kappa  )}{e(C_2,d_{A,2} \cdot d_{B,2}^\kappa  ) \cdot C_3^{d_{A,3} + \kappa
   d_{B,3}}} \Big)
 \end{eqnarray}
   with $\kappa=H(C_1,C_2,C_3)$. \\ \vspace{-0.3 cm}
 \item[Trace:] given an alleged private key $(d_{\ID,A},d_{\ID,B})$, with $d_{\ID,A}=(d_{A,1},d_{A,2},d_{A,3})$, for an identity $\ID$,
 check the validity of $d_{\ID}$ w.r.t. $\ID$
 using relations  (\ref{valid-priv-1-cca})-(\ref{valid-priv-2-cca}). If
 valid, the key is declared as a member of the family   $n_F=d_{3,A}=t_A$.
\end{description}
The proof of IND-sID-CCA security is omitted here as it is a
standard application of the
  technique  used in \cite{KV08}, which in turn borrows  ideas
from \cite{KD04,SG04,HK07}.

In the chosen-ciphertext scenario, the white-box FindKey security is
no longer unconditional but relies on the (weak) ciphertext
integrity property of the symmetric encryption scheme.
\begin{theorem} \label{find-cca}
The scheme is $\mathrm{FindKey}\textrm{-}\mathrm{CCA}$ secure
assuming the security of the key derivation function and the (weak)
ciphertext integrity of the symmetric encryption scheme. The
advantage of an adversary $\A$ making at most $q_d$ decryption
queries is bounded by
\begin{multline*}
\mathbf{Adv}_{\mathcal{A}}^{\mathrm{FindKey}\textrm{-}\mathrm{CCA}}(\lambda,\ell)
\leq     2  \cdot  q_d \cdot
\mathbf{Adv}^{\mathsf{CT\textrm{-}INT}}(\ell) \\   + 2 \cdot q_d
\cdot
 \mathbf{Adv}^{\mathsf{KDF}}(\lambda,\ell)      + \frac{ 2q_d^2 + q_d +1 }{
 p}. ~
\end{multline*}
\end{theorem}
\begin{proof}
Given in appendix \ref{proof-find-cca}. \qed
\end{proof}

\section{Authenticated Symmetric Encryption} \label{auth-sym}

A  symmetric encryption scheme is specified by a pair
$(\mathsf{E},\mathsf{D})$, where $\textsf{E}$ is the encryption
algorithm and $\textsf{D}$ is the decryption procedure, and a key
space $\mathcal{K}(\ell)$ where $\ell \in \mathbb{N}$ is a security
parameter. The security of authenticated symmetric encryption is
defined by means of two games that capture the ciphertext
indistinguishability and ciphertext (one-time) integrity properties.

\begin{definition} A  symmetric encryption scheme is  secure in the sense of authenticated encryption if
any PPT adversary has negligible advantage in the following games.
\end{definition}
\begin{enumerate}
\item \textbf{The \textrm{IND-SYM} game.} For  any PPT algorithm $\A$, the model considers the following game, where $\ell \in \N$ is a security parameter:
\begin{center}
\begin{tabular}{l}
\fbox{$\mathbf{Game}_{\A}^{\mathsf{IND\textrm{-}SYM}}(\ell)$}\\
$K \sample \mathcal{K}(\ell)$ \\
$(m_0,m_1,s) \leftarrow \A(\mathsf{find},\ell)$ \\
$d^\star \sample \{0,1\}$ \\
$c^\star \leftarrow \mathsf{E}_K(m_{d^\star})$ \\
$d \leftarrow \A(\mathsf{guess},s,c^\star)$ \\
\texttt{return} $1$ \texttt{if} $d=d^{\star}$ \texttt{and} $0$
\texttt{otherwise}.
\end{tabular}
\end{center}
$\mathcal{A}$'s  advantage  is
    $\mathbf{Adv}_{\mathcal{A}}^{\mathsf{IND\textrm{-}SYM}}(\ell) =
  | \Pr[\mathbf{Game}_{\mathcal{A}}^{\mathsf{IND\textrm{-}SYM}}=1] - {1}/{2}
  |.$ \\ \vspace{-0.2 cm}

\item[2.] \textbf{The \textrm{CT-INT} game.} Let $\A$ be a PPT algorithm. We consider the following game, where $\ell \in \N$ is a security parameter:
\begin{center}
\begin{tabular}{l}
\fbox{$\mathbf{Game}_{\A}^{\mathsf{CT\textrm{-}INT}}(\ell)$}\\
$K \sample \mathcal{K}(\ell)$ \\
$(m,s) \leftarrow \A(\mathsf{find},\ell)$ \\
$c \leftarrow \mathsf{E}_K(m)$ \\
$c'  \leftarrow \A(\mathsf{create},\ell,c)$ \\
\texttt{return} $1$ \texttt{if} $c' \neq c$ \texttt{and} $\mathsf{D}_K(c') \neq \perp$ \\
\phantom{{\texttt{return}}} $0$ \texttt{otherwise}.
\end{tabular}
\end{center}
$\mathcal{A}$'s advantage is now defined
  as $\mathbf{Adv}_{\mathcal{A}}^{\mathsf{CT\textrm{-}INT}}(\ell) =
 \Pr[\mathbf{Game}_{\mathcal{A}}^{\mathsf{CT\textrm{-}INT}}=1].$
\end{enumerate}

The notion of weak ciphertext integrity is defined in the same way
but the adversary is not allowed to see an encryption $c$ under the
challenge key $K$.

\section{Proof of Theorem \ref{find-cca}} \label{proof-find-cca}

The proof proceeds   with a sequence of two games, in   which
$S_i$ denotes the event that the adversary $\A$ wins during Game$_i$ with $i \in \{0,1\}$. \\
\vspace{-0.3 cm}

\noindent \textbf{Game}$_0$: is the FindKey-CCA experiment. The
dishonest PKG $\A$ generates the master public key, chooses an
identity $\ID$ that she wishes to be challenged upon. She interacts
with the challenger in a key generation protocol, upon completion of
which the challenger $\B$ obtains a
  decryption key consisting of two triples $d_{\ID,A}^{(1)} =
({d_{A,1}}^{(1)},{d_{A,2}}^{(1)},{d_{A,3}}^{(1)})$, $d_{\ID,B}^{(1)}
= ({d_{B,1}}^{(1)},{d_{B,2}}^{(1)},{d_{B,3}}^{(1)})$ that should
pass the key sanity check (otherwise, $\B$ aborts). At this stage,
$\A$ knows $t_B^{(1)}=d_{B,3}^{(1)}$ but    has no information on
$d_{A,3}^{(1)}=t_A^{(1)}$ or on the values
$r_A=\log_X(d_{A,2}^{(1)})$ and $r_B=\log_X(d_{B,2}^{(1)})$ (by the
construction of the key generation protocol). In the next phase,
$\A$ starts making a number of decryption queries that the
challenger handles using $(d_{\ID,A}^{(1)},d_{\ID,B}^{(1)})$.
Namely, when queried on a ciphertext $C=(C_1,C_2,C_3,C_4)$, $\B$
calculates
$$ \psi= \frac{e\big(C_1, d_{A,1}^{(1)} \cdot {d_{B,1}^{(1)}}^\kappa  \big)}{e\big(C_2,d_{A,2}^{(1)} \cdot {d_{B,2}^{(1)}}^\kappa  \big) \cdot
C_3^{{d_{A,3}^{(1)} + \kappa
   d_{B,3}}^{(1)}}}, $$
where $\kappa=H(C_1,C_2,C_3)$, $K=KDF(\psi)$ and
$m=\mathsf{D}_K(C_4)$ which is returned to $\A$ (and may be $\bot$ if $C$ is declared invalid).\\
\indent  At the end of the game, $\A$ outputs a key
$(d_{\ID,A}^{(2)},d_{\ID,B}^{(2)})$ and wins if $d_{\ID,A}^{(2)}$
parses into $({d_{A,1}}^{(2)},{d_{A,2}}^{(2)},{d_{A,3}}^{(2)})$ such
that
${d_{A,3}}^{(1)}=t_A^{(1)}=t_A^{(2)}={d_{A,3}}^{(2)}$.\\
\indent We note that decryption queries on well-formed ciphertexts
do not reveal any information to $\A$ (since all well-formed keys
yield the same result). We will show that, provided all ill-formed
ciphertexts are rejected by $\B$, $\A$ still has negligible
information on $t_{A}^{(1)}$ in the end of the game. For
convenience, we distinguish two types of invalid ciphertexts: type I
ciphertexts $(C_1,C_2,C_3,C_4)$ are such that  $\log_X (C_1) \neq
\log_{F(\ID)}(C_2)$ (and can be told apart from valid ones by
checking if $e(C_1,F(\ID)) \neq e(X,C_2)$), where $F(\ID)=g^\ID
\cdot Z$,
 whereas type II ciphertexts are those for which $\log_X (C_1) = \log_{F(\ID)}(C_2) \neq \log_{e(g,h)}
 (C_3)$. \\
\vspace{-0.3 cm}

\noindent \textbf{Game}$_1$: is as Game$_0$ but $\B$ rejects all
type I invalid ciphertexts (that are publicly recognizable). Such a
malformed ciphertext comprises elements $C_1=X^{s_1}$,
$C_2=F(\ID)^{s_1-s_1'}$  and
 $C_3=e(g,h)^{s_1-s_1''}$  where $s_1'>0$ and $s_1'' \geq 0$. Hence, the symmetric key
 $K$
 that $\B$ calculates is derived from
 \begin{eqnarray} \label{find-cca-eq}
   \psi
   &=&  e(g,Y_A^{s_1} \cdot Y_B^{\kappa  s_1}) \cdot e(F(\ID),X)^{s_1' (r_A+ \kappa r_B)} \cdot  e(g,h)^{s_1''(t_A^{(1)}  + \kappa t_B^{(1)} )}
\end{eqnarray}
where $\kappa=H(C_1,C_2,C_3 )$. Upon termination of the key
generation protocol, $\A$ has no information on $r_A,r_B$
 (as $\B$
re-randomizes its key). Even if  $\kappa$ was the same   in all
decryption queries (which may happen if these queries all involve
identical $(C_1,C_2,C_3)$),
  the second term of the product
(\ref{find-cca-eq}) remains almost uniformly random to $\A$ at each
new   query. Indeed, for each failed one, $\A$ learns at most one
value that is not $r_A+ \kappa r_B$. After $i$ attempts, $p-i$
candidates are left and the  distance between the uniform
distribution on $\G_T$ and that of $e(F(\ID),X)^{s_1' (r_A+ \kappa
r_B)}$ becomes at most $i/p \leq q_d /p$.  Then, the only way for
$\A$ to  cause the new rejection rule to apply  is to forge a
symmetric authenticated encryption for an essentially random   key
$K$. A standard argument shows that, throughout all
queries, the probability of $\B$ not rejecting a type I ciphertext
 is   smaller than $  q_d \cdot
(\mathbf{Adv}^{\textsf{CT-INT}}(\ell) +
\mathbf{Adv}^{\textsf{KDF}}(\lambda,\ell)  + q_d/p)$.  It easily
comes that $|\PR[S_1]-\PR[S_0]| \leq q_d \cdot
(\mathbf{Adv}^{\mathsf{CT\textrm{-}INT}}(\lambda) +
 \mathbf{Adv}^{\mathsf{KDF}}(\lambda,\ell)+q_d/p)$.\\ \noindent
 \indent  We now consider type II invalid
queries. While $\A$ knows $t_B^{(1)}$,  she has initially no
information on  $t_A^{(1)}$ and the last term of the product
(\ref{find-cca-eq}) is  unpredictable to her at the first type II
query. Each such  rejected   query allows $\A$ to rule out at most
one candidate as for the   value $t_A^{(1)}$. After $i\leq q_d$
unsuccessful type II queries, she is left with at least $p-i$
candidates  at the next type II query, where the distance between
the uniform distribution on $\G_T$ and that of $\psi$ (calculated as
per (\ref{find-cca-eq})) becomes smaller than $i/p \leq q_d/p$.
Again, one can show that, throughout all queries, the probability of
$\B$ not rejecting a type II ciphertext
 is at most $  q_d \cdot
(\mathbf{Adv}^{\textsf{CT-INT}}(\ell) +
\mathbf{Adv}^{\textsf{KDF}}(\lambda,\ell)  + q_d/p)$. Let us call
$\texttt{type-2}$ the latter event.    If all invalid ciphertexts
are rejected, $\A$'s probability of success is given by
$\PR[S_1|\neg\texttt{type-2}]\leq 1/(p-q_d) \leq (q_d+1)/p$.
 Since
\begin{eqnarray*}
\PR[S_1] & =&  \PR[ S_1 \wedge \texttt{type-2} ] + \PR[S_1 \wedge
\neg \texttt{type-2} ] \\ & \leq & \PR[ \texttt{type-2} ]+ \PR[
S_1|\neg\texttt{type-2} ] \PR[\neg \texttt{type-2}] \\
& \leq & \PR[ \texttt{type-2} ] +  \PR[ S_1|\neg\texttt{type-2} ] \\
& \leq &  q_d \cdot \big(\mathbf{Adv}^{\textsf{CT-INT}}(\ell) +
\mathbf{Adv}^{\textsf{KDF}}(\lambda,\ell)  + \frac{q_d}{p} \big) +
\frac{q_d+1}{p}
\end{eqnarray*}
and $|\PR[S_0]-\PR[S_1]| \leq
  q_d \cdot
(\mathbf{Adv}^{\mathsf{CT\textrm{-}INT}}(\lambda) +
 \mathbf{Adv}^{\mathsf{KDF}}(\lambda,\ell)+q_d/p) $, the claimed upper
 bound follows. \qed

\end{document}